\documentclass[12pt,journal,cspaper,compsoc,onecolumn]{IEEEtran}

\ifCLASSINFOpdf
   \usepackage[pdftex]{graphicx}
   \DeclareGraphicsExtensions{.pdf,.jpeg,.png}
\else
   \usepackage[dvips]{graphicx}
   \DeclareGraphicsExtensions{.eps}
\fi

\usepackage[cmex10]{amsmath}

\usepackage[linesnumbered,ruled,vlined]{algorithm2e}

\ifCLASSOPTIONcompsoc
\usepackage[tight,normalsize,sf,SF]{subfigure}
\else
\usepackage[tight,footnotesize]{subfigure}
\fi

\usepackage{caption}

\usepackage{multirow}
\usepackage{slashbox}
\usepackage{array}

\usepackage{setspace}
\usepackage{amsfonts}
\usepackage{amsthm}

\newcommand{\sbt}{\,\begin{picture}(-1,1)(-1,-3)\circle*{3}\end{picture}\ }

\begin{document}
%
\title{A Real-time Group Auction System for Efficient Allocation of Cloud Internet Applications}
%
%
%
%

\author{Chonho Lee, \ Ping Wang, \ Dusit Niyato
\thanks{C. Lee, P. Wang, and D. Niyato are with the School of Computer Engineering at Nanyang Technological University in Singapore. Email: \{leechonho, wangping, and dniyato\} @ ntu.edu.sg}}

%
%

\markboth{IEEE Transactions on Services Computing, ~Vol.~xx, No.~xx}%
{Shell \MakeLowercase{\textit{et al.}}: Bare Demo of IEEEtran.cls for Computer Society Journals}

\IEEEcompsoctitleabstractindextext{%
\begin{abstract}
Increasing number of the cloud-based Internet applications demands for efficient resource and cost management.
This paper proposes a real-time group auction system for the cloud instance market. The system is designed based on a combinatorial double auction, and its applicability and effectiveness are evaluated in terms of resource efficiency and monetary benefits to auction participants (e.g., cloud users and providers).
The proposed auction system assists them to decide when and how providers allocate their resources to which users.
Furthermore, we propose a distributed algorithm using a group formation game that determines which users and providers will trade resources by their cooperative decisions.
To find how to allocate the resources, the utility optimization problem is formulated as a binary integer programming problem,
and the nearly optimal solution is obtained by a heuristic algorithm with quadratic time complexity.
In comparison studies, the proposed real-time group auction system with cooperation outperforms an individual auction in terms of the resource efficiency (e.g., the request acceptance rate for users and resource utilization for providers) and monetary benefits (e.g., average payments for users and total profits for providers).
\end{abstract}

}

\maketitle

\IEEEdisplaynotcompsoctitleabstractindextext

%
\IEEEpeerreviewmaketitle

\section{Introduction}

\newcommand{\nbiafWin}{53}
\newcommand{\nbiafLose}{21}

\newcommand{\nbgafWin}{59}
\newcommand{\nbgafLose}{15}

\newcommand{\bgaWin}{56}
\newcommand{\bgaLose}{18}

\newcommand{\gabctWin}{58}
\newcommand{\gabctLose}{16}

\newcommand{\gabctgfWin}{60}
\newcommand{\gabctgfLose}{14}

\newcommand{\bgacopWin}{67}
\newcommand{\bgacopLose}{7}

\newcommand{\nbiafProvone}{\$236.3}
\newcommand{\nbiafProvtwo}{\$145.8}
\newcommand{\nbiafProvSum}{\$382.1}

\newcommand{\nbgafProvone}{\$271.1}
\newcommand{\nbgafProvtwo}{\$192.4}
\newcommand{\nbgafSum}{\$463.5}

\newcommand{\bgaProvone}{\$195.5}
\newcommand{\bgaProvtwo}{\$281.3}
\newcommand{\bgaSum}{\$476.8}

\newcommand{\gaProvone}{\$186.2}
\newcommand{\gaProvtwo}{\$208.4}
\newcommand{\gaProvSum}{\$394.6}

\newcommand{\gagfProvone}{\$191.7}
\newcommand{\gagfProvtwo}{\$212.1}
\newcommand{\gagfProvSum}{\$403.8}

\newcommand{\bgacopProvone}{\$218.7}
\newcommand{\bgacopProvtwo}{\$345.8}

\newcommand{\iaAvgCost}{\$7.21}
\newcommand{\gaAvgCost}{\$6.84}
\newcommand{\gagfAvgCost}{\$6.73}

\newcommand{\nbiafResutil}{70.5\%}
\newcommand{\nbgafResutil}{77.3\%}
\newcommand{\bgaResutil}{79.5\%}
\newcommand{\gaResutil}{82.8\%}
\newcommand{\gagfResutil}{88.7\%}
\newcommand{\bgacopResutil}{94.8\%}

\newcommand{\iaDone}{2224} \newcommand{\iaDtwo}{1831} \newcommand{\iaDthree}{2035}
\newcommand{\gaDone}{2655} \newcommand{\gaDtwo}{2102} \newcommand{\gaDthree}{2398}
\newcommand{\gagfDone}{2942} \newcommand{\gagfDtwo}{2195} \newcommand{\gagfDthree}{2531}

\newcommand{\nbiafMargin}{\$11.2}
\newcommand{\nbgafMargin}{\$16.5}
\newcommand{\bgaMargin}{\$14.5 }
\newcommand{\bgacoMargin}{\$18.4 }
\newcommand{\bgacopMargin}{\$25.1 }

\newcommand{\upReqAR}{13\% }		
\newcommand{\upResUtil}{26\% }	
\newcommand{\upPay}{7\% }				
\newcommand{\upProfit}{6\% }		

\newcommand{\CInbiafWin}{53.57 $\pm$ 4.24}
\newcommand{\CIgabctWin}{59.85 $\pm$ 7.78}
\newcommand{\CIgabctgfWin}{62.88 $\pm$ 4.95}

\newcommand{\CInbiafResutil}{69.26\% $\pm$ 2.83\%}
\newcommand{\CIgaResutil}{81.01\% $\pm$ 2.83\%}
\newcommand{\CIgagfResutil}{85.24\% $\pm$ 4.24\%}

\newcommand{\CInbiafProvone}{\$242.4 $\pm$ \$32.1}
\newcommand{\CInbiafProvtwo}{\$175.1 $\pm$ \$42.3}
\newcommand{\CInbiafProvSum}{\$394.1 $\pm$ \$20.2}
\newcommand{\CIgaProvone}{\$201.8 $\pm$ \$28.8}
\newcommand{\CIgaProvtwo}{\$223.4 $\pm$ \$47.6}
\newcommand{\CIgaProvSum}{\$429.3 $\pm$ \$71.3}
\newcommand{\CIgagfProvone}{\$202.5 $\pm$ \$22.1}
\newcommand{\CIgagfProvtwo}{\$238.7 $\pm$ \$52.1}
\newcommand{\CIgagfProvSum}{\$442.4 $\pm$ \$67.5}

Cloud computing is a large-scale distributed computing leveraging Internet-accessible data centers
that provide computing resources (e.g., CPU, memory, storage, and network bandwidth) as cloud.
Modern Internet applications are designed using the virtualization technology in the cloud computing environment.
Such cloud-based Internet applications are deployed on virtual machines (VMs), also called {\it instances}~\cite{AmazonEC2} whose resource requirements are pre-configured.
Cloud users (e.g., application developers) request a certain number of instances to run their applications, which are hosted by cloud providers (e.g., owners of physical machines).
According to the number of instances that cloud users use, cloud providers charge the fee for what the cloud users used.

This paper studies the issues of cost efficiency for cloud users and resource efficiency for cloud providers.
Users are required to minimize the cost to deploy and run their applications on clouds and to maximize the reliability of the applications by improving the resource availability.
On the other hand, providers are required to maximize their revenue by attracting users and improving the resource utilization.
However, it is hard for users to determine which providers that they should trade with, and vice versa, to meet their objectives under dynamic and unpredictable resource demands, supplies, prices, and budgets.

The goal of this paper is to design a real-time auction system for the cloud instance market where multiple cloud users and providers respectively buy and sell instances to run and host cloud-based Internet applications.
The proposed auction system assists the auction participants (i.e., both cloud users and providers) to decide how providers allocate their resources to which users in order to improve both cost and resource efficiency.

To achieve the goal, we introduce a group auction and the participants' cooperation.
A group auction is a type of auctions and a market method that benefits both buyers and sellers. Specifically, sellers offer buyers price discounts according to the number of requested items (i.e., the more the requested items, the lower the price).
One example scenario of such auction is the online group-buying shopping sites \cite{eWinWin}, \cite{groupgain}, and \cite{groupon}, which have recently become popular.
However, the group auction has not been extensively analyzed in cloud computing although it has been well studied in economics.
One of the objectives of this paper is to investigate the applicability and effectiveness of a group auction to the cloud instance market.

In the proposed group auction, multiple cloud users are allowed to form a group, called a {\it user coalition} in order to buy their instances at a discounted price.
Multiple cloud providers are also allowed to form a group, called a {\it provider coalition} and share their resources to effectively use the residual resources, which results in the improvement of resource utilization.

In such a cooperative system, users expect to buy instances at a low price, and the providers expect to attract more users to improve the resource utilization and thus to increase their revenues.
However, the cooperation may cause a delay of the instance allocation (e.g., by waiting for more users to join an auction) and a cost for the instance migration~\cite{Voorsluys09} (e.g., a cost incurred by the performance degradation of applications when instances are migrated to different servers in different providers).
There is a need for the analysis of participants' cooperation decisions, i.e., when to or not to cooperate.
It is challenging to find stable set of such decisions, which satisfies all of rational participants trying to maximize only their own benefits.
The contributions of this paper can be summarized as follows:
\begin{itemize}
\item
We propose a real-time group auction system in the cloud instance market, which is designed based on a combinatorial double auction,
and verify its applicability and effectiveness
in terms of resource efficiency and benefits to auction participants.
\item
We propose a distributed algorithm using a group formation game that determines which users and providers will trade resources by their cooperative decisions.
\item
The utility optimization problem is formulated as a binary integer programming problem for the resource allocation,
and the nearly optimal solution is obtained by a heuristic algorithm with quadratic time complexity.
\end{itemize}

The proposed real-time group auction system gives auction participants the resource efficiency (e.g., \upReqAR higher request acceptance rate for users and \upResUtil higher resource utilization for providers) and monetary benefits (e.g., \upPay lower average payments for users and \upProfit higher total profits for providers) compared to a conventional individual auction with the best-fit allocation.
The proposed group formation and instance allocation algorithms are analyzed in their complexity, stability, and optimality. The results show that the algorithms with quadratic time complexity always obtain a nearly optimal group structure regardless of configurations and initial conditions.
The proposed system acts as a glue that binds cloud users and providers, and it can be used by providers or third parties (i.e., brokers) for efficient cloud service delivery and resource management in the practical system.

The remainder of this paper is organized as follows:
\begin{itemize}
\item
Section~\ref{sec:relatedwork} introduces existing related work and summarizes their similarities and differences from our work.
\item
Section~\ref{sec:systemmodel} explains how the cloud instance market works in a group auction manner, and then overviews a system model for the proposed group auction. Some assumptions that we considered in the system are stated in this section.
\item
Section~\ref{sec:groupauctionsystem} formulates the instance allocation as the social welfare optimization and presents in detail how to achieve the nearly optimal solution by the proposed group formation algorithm. 
\item
Section~\ref{sec:analysis} gives a theoretical analysis of the proposed group formation algorithm in terms of the time complexity and stability.
\item
Section~\ref{sec:evaluation} discusses the convergence and optimality of the proposed algorithm on numerical results and compare the results with different schemes through the simulation, which is followed by a conclusion.
\end{itemize}

\section{Related Work}
\label{sec:relatedwork}

This section introduces existing related work and describes their similarities and differences from our work, which are summarized in Table~\ref{tab:relatedwork} in terms of resource allocation methods, auction types, and cooperation methods.


\vspace{+6pt}
\noindent {\bf Resource Allocation}

The importance of resource provisioning has been well discussed in various fields such as wireless networks, energy industries, and advertisements, which have proposed the allocation and pricing model of resources (e.g., wireless channels~\cite{Berry10,Mutlu09}, electricity~\cite{Markus04,Kian05}, and advertisements~\cite{Bhalgat11,Mahdian07}) to improve the resource utilization and efficiency.
We focus on instances in clouds and consider an instance market where computing resources (e.g., bandwidth, CPU time and memory space) are traded as instances.

For the resource allocation, there exist several techniques such as
game theory finding an equilibrium solution among players~\cite{Wei10,Ardagna11};
stochastic programming considering uncertainty~\cite{Chaisiri11};
and bio-inspired mechanisms (e.g., genetic algorithm that seeks a Pareto solution of a multi-objective problem~\cite{Singh06} and Ant colony that provides a heuristic solution of a complex problem~\cite{Chimakurthi11}).
We apply the auction theory to design the cloud instance market and formulate its instance allocation.

\vspace{+6pt}
\noindent {\bf Auction-based Mechanism}

Auction-based mechanisms have been proposed in various fields such as wireless networks and cloud computing in order to investigate how participants (or nodes) behave in a competition for resources; and different classes of auctions such as sequential second price auction~\cite{JBae08}, Vickrey auction~\cite{Stanojev10}, double auction~\cite{TanGurd07}, and combinatorial auction~\cite{Fujiwara10} have been considered in the design of the mechanisms.
However, none of them considers a group auction and the participants' cooperation.
We investigate a combinatorial double auction executed in a group-buying manner and analyze the optimal allocation by observing participants' cooperative decisions in a group (or coalition) formation game.


\vspace{+6pt}
\noindent {\bf Cooperation}

Cooperation is one of the important concepts to improve the resource efficiency.
The distributed resource management schemes in computational grid were developed
with negotiation algorithm~\cite{Linli05} and coalition formation algorithm \cite{GongYong03}.
These schemes allow rational agents managing server farms to form the cooperation to share available resources.
The cooperative task scheduling has been proposed in~\cite{Pascual09}.
It was shown that it is always possible to obtain collaborative solution of the self-interested agents which can improve the global system performance.
We adopt the concept of cooperation among users and providers under the consideration of both the cost and resource efficiency.
None of the works considered a dynamic instance allocation in a group auction from the coalitional game perspective.

\begin{table}[t]
\small
  \caption{Summary of related work ({\small A symbol "-" indicates omission for out of the paper scope.})}
	\label{tab:relatedwork}
  \centering
	\renewcommand{\arraystretch}{1.2}
	\begin{tabular}{|l|c|c|c|c|c|c|}
	\cline{1-3}\cline{5-7}
	\multirow{2}{*}{{\bf Resources}} & \multicolumn{2}{c|}{{\bf Allocation methods}} &&
	{\bf Auction types}	& Single-sided & Double-sided	
	\\	
	\cline{2-3}\cline{5-7}
	 			& Auction-based & Others && Individual & [24] &	[19, 20, 25]
	\\	
	\cline{1-3}\cline{5-7}
	Electricity & [8, 9] & - &&  Group & [2, 3, 4] & Our work
	\\	
	\cline{1-3}\cline{5-7}
	Wireless spectrum & [6, 17, 18] & - &\multicolumn{1}{c}{}&
	\multicolumn{3}{c}{}
	\\	
	\cline{1-3}\cline{5-7}
	\multirow{2}{*}{Cloud instances} & Our work & [12, 14, 15 &&
	\multirow{2}{*}{{\bf Cooperation}}  & Centralized & Self-interested
  \\	
					 & [19, 20, 24, 25] & 16, 21, 22] &&
	         & decision & decision 
  \\	
	\cline{1-3}\cline{5-7}
	Others	 & [2, 3, 4]       & - &&
	Buyers &  [22] & [21]
  \\	
	\cline{1-3}\cline{5-7}
	\multicolumn{3}{c}{} &&
  Buyers and/or Sellers   & - & Our work 
  \\ 	
	\cline{5-7}
	\end{tabular}
\end{table}

\vspace{+6pt}
\noindent {\bf Cloud Instance Market}

The proposed work differs from the major cloud hosting services such as AmazonEC2's Spot Instance~\cite{AmazonEC2-SpotInstance} and SpotCloud~\cite{SpotCloud}, which deal with the instance market.
\cite{AmazonEC2-SpotInstance} tries to sell the residual resources to cloud users to achieve high resource utilization. Users join an auction to reserve instances and pay for the resources at a dynamically changing spot price offered based on the supply-demand conditions.
Our work is similar to this in terms of the trading price changed based on supply-demand conditions.
But, \cite{AmazonEC2-SpotInstance} is considered to be a one-sided auction that users' requests are individually processed in a first-come-first manner. Our work is a double auction executed in a group auction manner.

Similar to the proposed system, \cite{SpotCloud} works as a double auction, and resource prices change depending on providers' valuations to the resources.
Users submit their instance requirements, and providers submit profiles of hardware to be supplied.
\cite{SpotCloud} lists the profiles best-matched for the user requirements, and the users will choose what they want from the list.
For providers, \cite{SpotCloud} automatically creates instances on providers' machines requested by users.
Our proposed approach can be applied to such SpotCould market to automate the users' provider selection step; it tries to find the best matching of users and providers cooperating each other to gain benefits in terms of monetary cost and resource efficiency.

In the evaluation section, we verify the impact of group auction and cooperation by comparing the proposed scheme with the two schemes. However, for the fair comparison, we extend them as follows, and we respectively call IA-FCFS and GA-BCT (described in Table 4). 
IA-FCFS mimics \cite{AmazonEC2-SpotInstance} but considering multiple providers, and we try to observe the impact of group auction.
GA-BCT mimics \cite{SpotCloud} but following our allocation algorithm, and we try to observe the impact of cooperation among users and providers.


\section{The System Model and Assumptions}
\label{sec:systemmodel}

\subsection{Cloud Instance Market}
\label{subsec:cloudinstancemarket}

We consider the system model for the cloud instance market where multiple cloud users and providers respectively buy and sell instances in a group auction manner. To buy/sell instances, the users/providers submit bids/offers to a central controller which supports the instance allocation and pricing determination.
This paper considers a discrete-time (e.g., 15 minutes or half an hour) system so a concept of {\it time slot} is involved.
We say time slot $t$ for a period between time $t-1$ and time $t$.
For example, saying that the system collects 3 bids at time slot 5 means that 3 bids are submitted during a period between time instants 4 and~5.

The system deals with $K$ different types of instances.
The instance in different types will require the different size of computing resources such as memory, storage, CPU capacity (e.g., EC2 Compute Unit\footnote{One EC2 Compute Unit provides the equivalent CPU capacity of a 1.0-1.2 GHz 2007 Opteron or 2007 Xeon processor.}), and network bandwidth associated with I/O performance (e.g., throughput and latency).
For example, Amazon EC2~\cite{AmazonEC2} offers various types of Standard Instance 
shown in Table~\ref{tab:instancetype}.
Different providers provide various types of instances with different combinations of computing resources.

\begin{table}[h]
  \caption{Types of Standard Instance Family in Amazon EC2}
	\label{tab:instancetype}
  \centering
	\begin{tabular}{|c||c|c|c|c|}
	\hline
	\bf Instance type & \bf Small & \bf Medium & \bf Large & \bf Ex-large  \\	
	\hline
	Memory  & 1.7 GB & 3.75 GB & 7.5 GB & 15 GB   \\	
	\hline
	Storage & 160 GB & 410 GB & 850 GB &  1690 GB \\	
	\hline
	EC2 Compute Unit & 1 & 2 & 4 & 8 \\	
	\hline
	I/O performance  & Moderate & Moderate & High &  High \\	
	\hline
	\end{tabular}  
\end{table}
\noindent

Modern Internet applications are designed using multiple servers at multiple tiers, each of which provides a certain functionality \cite{Urgaonkar07}.
For example, a front-end Web server is responsible for HTTP processing, a middle-tier Java application server implements the application logic, and a back-end database server stores user information.
The different tiers of an application are assumed to be distributed across different servers.
Depending on the desired capacity, a tier may also be clustered by multiple servers of the same type, which are connected in parallel.

The servers are requested/provided by users/providers as {\it instances} \cite{zaman2010,Daniel2011} mentioned above.
A cloud user requests a bundle of instances that represents a multitier application and buys the instances in a taking-all-or-none manner.
For example, a user requests four Small instances (for Web servers), two Large instances (for application servers), and one Ex-large instance (for a database).
If one of the requested instances cannot be provided, then other requested instances become useless to the user.
We also assume that the bundle of instances will be hosted by one provider to take advantage (e.g., high throughput and short latency) of the proximity among servers for an application.

\subsection{Participants of the Cloud Instance Market}
\label{subsec:participants}

\subsubsection{Cloud Users}
A cloud user $i$ submits a bid defined by $b_i = (\vec{d_i}, \vec{\ell_i}, v_i )$ where $\vec{d_i}=(d_i^1, d_i^2, \ldots, d_i^K)^T$ is a {\it demand}, and $d_i^k$ indicates the number of instances of type $k$.
$\vec{\ell_i}=(\ell_i, t_i^{s},t_i^{e})^T$ is a {\it demand period} where $\ell_i$ indicates a {\it length} of time that the user $i$ wants to reserve a bundle of the instances between {\it starting time} $t_i^{s}$ and {\it ending time} $t_i^{e}$.
For example, $\vec{\ell_i}=(3,2,8)$ indicates that user $i$ requests instances for 3 time slots (that do not have to be consecutive) between time 2 and time 8, each of the time slots reserves $\vec{d_i}$ instances.
We define $t_i^d \equiv t_i^e-\ell_i$ and call it the {\it deadline} by which the instances should be assigned to satisfy the demand. 
$v_i$ is user $i$'s {\it valuation} for the demand as a bidding price, which indicates the maximum price that is acceptable for the user to buy the requesting instances.

\subsubsection{Cloud Providers}
A cloud provider $j$ submits an offer defined by $o_j = (\vec{s_j}, \vec{w_j}, Q_j)$ where $\vec{s_j}=(s_j^1, s_j^2, \ldots, s_j^K)^T$ is a {\it supply} and $s_j^k$ indicates the number of instances of type $k$ that provider $j$ can provide per time slot.
$\vec{w_j}=(t_j^{s},t_j^{e})^T$ is a {\it supply period}\footnote{We formalize the supply period as variables whose values are different for different providers. The providers might be personal providers or small companies that supply their resources for limited time periods~\cite{SpotCloud} while large providers \cite{AmazonEC2} do not limit the supply period in general.}.
$w_j \equiv t_j^{e}-t_j^{s}$ denotes a {\it length} of time that a provider is able to provide the instances between $t_j^{s}$ and $t_j^{e}$.
$Q_j$ is provider $j$'s valuation\footnote{In practice, the provider's valuation to an instance is determined based on several factors such as the value of applications to users, the running cost of servers, and other competing providers. It is beyond our work in this paper to understand how providers set their own valuations. For simplicity, we assume that the valuations are pre-determined.} for the supply as an offering {\it price curve}, which indicates the minimum of a unit price of the offered instances that the provider wishes to sell.
It is defined by a set of vectors over different number of instances and instance types as follows:
$Q_j = (\vec{q}_j^1, \vec q_j^2, \ldots, \vec q_j^K)$
where $\vec q_j^k=(q_j^k[1],\ldots,q_j^k[n_j^k])^T$, and $q_j^k[n]$ indicates the offering price when $n$ instances of type $k$ are sold by provider $j$, and $q_j^k[0] = +\infty$.
$\vec{q}_j^k$ holds a condition $q_j^k[1] \geq \cdots \geq q_j^k[n_j^k]$.
Let $Q_j[\vec{n_j}]$ denote an extraction of the price curve when $\vec{n_j}=(n_j^1,n_j^2,\ldots,n_j^K)$ instances are sold, which is represented as a row vector
$Q_j[\vec{n_j}] = (q_j^1[n_j^1], q_j^2[n_j^2],\ldots, q_j^K[n_j^K])$.

\subsection{Overview of the Proposed Group Auction System}
\label{subsec:overview}

\begin{figure}[t]
  \centering
	\resizebox{0.95\linewidth}{!}{\includegraphics{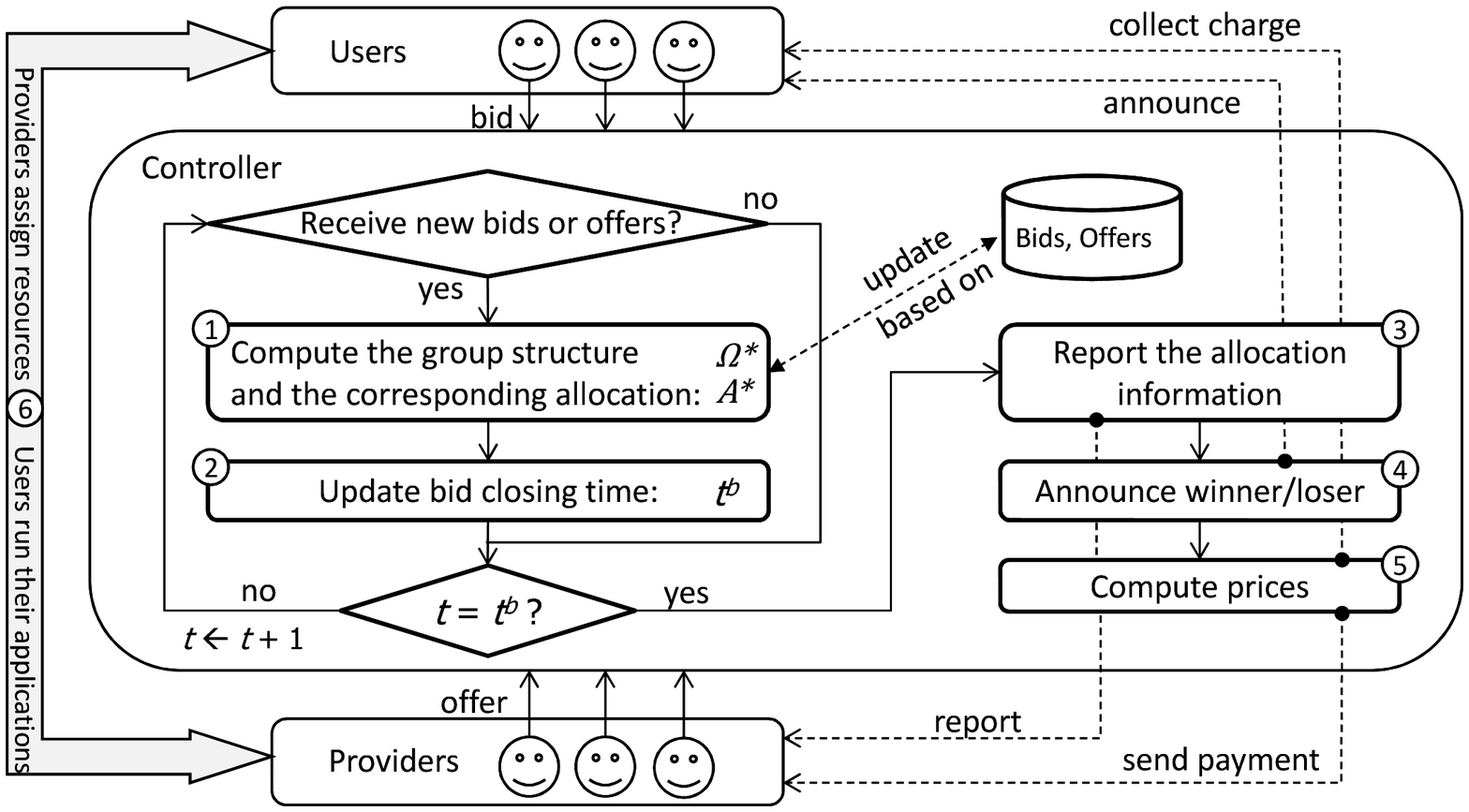}}
  \caption{A flowchart of the proposed group auction system.}
  \label{fig:flowchart}
\end{figure}

The central controller maintains the bids and offers collected from cloud users and providers respectively, and computes {\it when} and {\it how} to allocate resources to {\it which} users.
Figure~\ref{fig:flowchart} is a flowchart of the proposed group auction system and shows how the central controller works.
The system undertakes three main tasks such as 
the {\it allocation computation} (Labels \raisebox{.5pt}{\textcircled{\raisebox{-.9pt}{1}}} and \raisebox{.5pt}{\textcircled{\raisebox{-.9pt}{2}}}),
the {\it result reporting} (Labels \raisebox{.5pt}{\textcircled{\raisebox{-.9pt}{3}}} and \raisebox{.5pt}{\textcircled{\raisebox{-.9pt}{4}}}), and
the {\it payment management} (Label \raisebox{.5pt}{\textcircled{\raisebox{-.9pt}{5}}}).

Whenever the central controller receives new bids/offers, it updates their information on database and computes the best instance allocation (Label \raisebox{.5pt}{\textcircled{\raisebox{-.9pt}{1}}}) adopting a concept of cooperation among users and providers (described in the following section).
If the controller closes the bid/offer submission (Label \raisebox{.5pt}{\textcircled{\raisebox{-.9pt}{2}}}), it reports the recently computed instance allocation to providers (Label \raisebox{.5pt}{\textcircled{\raisebox{-.9pt}{3}}}) and announces to users who are the winners/losers of the auction (Label \raisebox{.5pt}{\textcircled{\raisebox{-.9pt}{4}}}).
Once the charging and payment are complete (Label \raisebox{.5pt}{\textcircled{\raisebox{-.9pt}{5}}}), users and providers establish the connection and start to run/host applications (Label \raisebox{.5pt}{\textcircled{\raisebox{-.9pt}{6}}}).

The instance allocation is formulated as a social welfare (or utility) optimization problem,
and the solution is obtained by the proposed group formation algorithm.
The trading prices are determined based on the social welfare distribution scheme.
The details are presented in Section \ref{sec:groupauctionsystem}.
The formulations of both allocation and pricing are considered to satisfy three auction properties as follows:

$\sbt$ An allocation is {\it allocatively efficient} if there are no participants who gain utility from decreasing others' utility.

$\sbt$ An allocation is {\it individually rational} if participants are never charged more than their valuations as a result of the allocation.

$\sbt$ An allocation is {\it budget-balanced} if the total profits of providers is the same as the total payments by users.

\subsubsection{Cooperations among Users and among Providers}
\label{subsubsec:cooperation}

In the proposed system, users and providers can form coalitions to gain benefits in terms of resource price and resource utilization efficiency.
Figure~\ref{fig:example_cooperation} shows two examples explaining two types of cooperation.
One, illustrated in Figure 2(a), is the case that multiple users request resources of the same provider to lower the instance price. 
Another, illustrated in Figure 2(b), is the case that multiple providers supply resources to meet the same user's request.
An edge-dotted line indicates the allocation of a user and a provider.
A shaded box indicates the allocated instances. The height of the box represents the number of instances.

\begin{figure}[t]
  \centering
	\resizebox{0.9\linewidth}{!}{\includegraphics{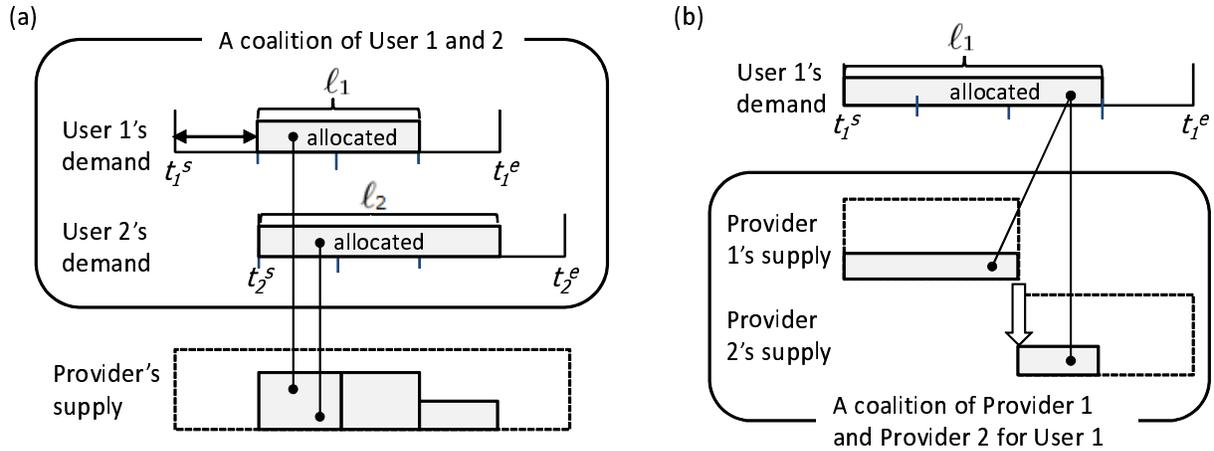}}
  \caption{An example allocation when (a) users cooperate and (b) providers cooperate. The shaded box indicates the allocated instances. The height of the box represents the number of instances.}
  \label{fig:example_cooperation}
\end{figure}

For the first type, multiple cloud users are allowed to form a group, i.e., we call a {\it user coalition}, and request instances of the same provider in order to buy the instances at a discounted price.
For example, when two users request 10 instances ($d_1=4$ and $d_2=6$) as a coalition, a provider's offering price changes from $q[4]$ (or $q[6]$) to $q[10]$ where $q[4]=q[6] > q[10]$.
However, a cooperation may incur a delay of the instance allocation because different users have different demand periods.
Figure~\ref{fig:example_cooperation}(a) explains the delay (denoted by an arrow) of User 1 forming a coalition with User 2.
Instances for User 1 could be allocated from time~$t_1^s$.

For the second type, multiple cloud providers are also allowed to form a group, i.e., we call a {\it provider coalition}.
The providers in the same coalition can share users' demands to meet their requests.
For doing this, users' instances can be migrated from one provider to another as indicated by a down arrow in Figure~\ref{fig:example_cooperation}(b).
This cooperation enables the effective use of the residual resources and improves their resource utilization.
However, the cooperation may incur a cost for the instance migration.
For example, during migration, the degradation of application performance (e.g., throughput and response time) can be considered as the cost.

\subsubsection{Allocation and Bid Closing Time}

Denote $B_t$ as a set of bids and $O_t$ as a set of offers that the controller retains at time $t$, respectively.
A {\it demand period} for demands in $B_t$ is specified by $\underline{t} = \min_{} \{t_i^s | t_i^s \in b_i, \forall b_i \in B_t\}$ and $\overline{t}  = \max_{} \{t_i^e | t_i^e \in b_i, \forall b_i \in B_t\}$.
The {\it allocation} at time $t$ is defined by a three-dimensional ($|B_t|\times|O_t|\times(\overline{t}-\underline{t})$) matrix
$A_t \equiv [a_{ijs}]$
for $i=1,\ldots,|B_t|$, $j=1,\ldots,|O_t|$, $s=\underline{t}+1,\ldots,\overline{t}$
where a binary variable $a_{ijs}$ is 1 if user $i$'s demand is allocated to provider $j$ at time slot $s$, and 0 otherwise.
We assume that a user's demand is satisfied when all of the numbers of requested instances is allocated. Also, we assume that user's instances are allocated to one provider at each time slot.
Thus, the system does not consider the partially allocated instances.
We denote the optimal allocation by $A_t^* \equiv [a_{ijs}^*]$, where the method to obtain it will be discussed in the next section.

Given the allocation $A_t$, the controller computes the best time to announce the allocation information (denoted by in Figure 1).
We call it {\it bid closing time},
which is denoted as $t^b = \min_{A_t} \{t_i^a\}$ where $t_i^a$ is the earliest time that user $i$'s demand is allocated according to $A_t$.
%
Figure~\ref{fig:example_allocation} illustrates an example of allocation for three users and one provider over the demand period of eight time slots.
In this example, the allocation is
$A_t=[\vec{a}_{11s},\vec{a}_{21s},\vec{a}_{31s}]=[(0,0,0,1,1,1,0,0)^T,(0,0,0,1,1,1,1,0)^T,(0,0,1,1,1,0,0,0)^T]$,
the demand period is $\underline{t}=t_3^s$ and $\overline{t}=t_2^e$, and
the bid closing time is $t^b=t_3^a$.
\begin{figure}[t]
  \centering
	\resizebox{0.75\linewidth}{!}{\includegraphics{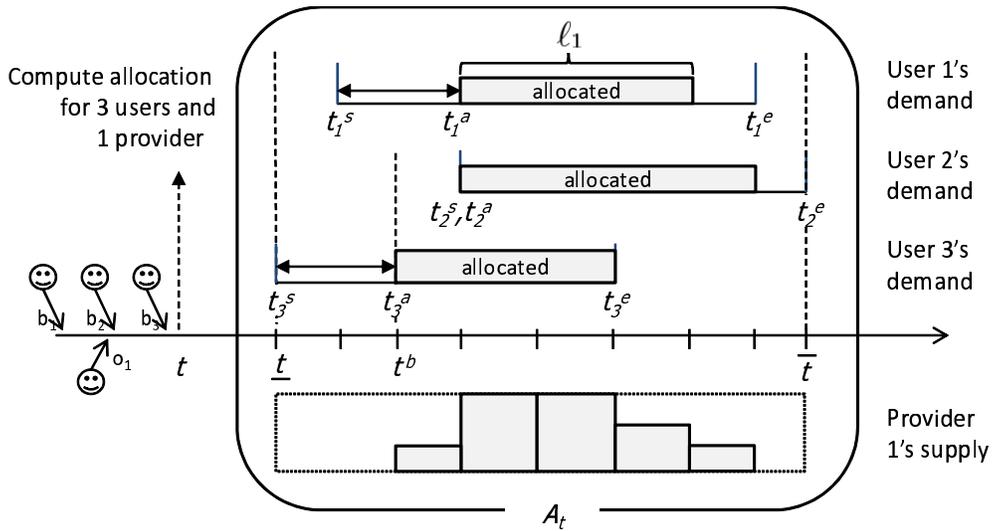}}
  \caption{An example of allocation in a group consisting of 3 users and 1 provider over the demand period of 8 time slots.}
  \label{fig:example_allocation}
\end{figure}

\section{The Proposed Group Auction System}
\label{sec:groupauctionsystem}

The instance allocation is formulated as a social welfare optimization problem.
The social welfare is iteratively improved by the proposed group formation algorithm that allows participants to make {\it group formation decisions} to increase their own payoffs.
This section first explains the group and the group structure defined in this paper (Section~\ref{subsec:groupstructure}) and presents in details instance allocation formulation (Section~\ref{subsec:instanceallocation}) and the method to obtain the nearly optimal solution by the group formation algorithm (Section~\ref{subsec:groupformationgame}).

\subsection{Group Structure}
\label{subsec:groupstructure}

A {\it group} denoted by $G$ consists of a set of users $G^u$ and a set of providers $G^p$, i.e., $G=G^u \cup G^p$.
Note that $G^u$ can be an empty set while $G^p$ is a non-empty set (i.e., including at least one provider).
Instances supplied by the providers are allocated to the users in the same group. 
Within a group, the users cooperate and get a discount when their instances of the same type are supplied by the same provider, and the providers are allowed to migrate instances among cooperating providers.

A {\it group structure} denoted by $\Omega_t$ is represented as a partition $\Pi_t$ with a set of associated links $L_t$ between users and providers.
The partition is a set of groups defined by $\Pi_t = ( G_0, G_1, \ldots, G_{m}, \ldots, G_M )$
where $G_m = \{ G^u_m, G^p_m\}$, the $m$-th group of users $G_m^u$ and providers $G_m^p$.
$G_0$ indicates a group of users without provider. The users in this group will not join an auction immediately due to resource limitation (i.e., they are likely to lose an auction) but wait for the next auction opportunity because the deadline has not been reached yet.
Denote $I_t$ as a set of all users with the bids $B_t$ and $J_t$ as a set of all providers with the offers $O_t$.
The partition holds following conditions $\cup_{m=0}^M G_m^u=I_t$ and $G_m^u \cap G_{m'}^u = \emptyset$ for all $m \neq m'$, and $\cup_{m=1}^M G_m^p=J_t$ and $G_m^p \cap G_{m'}^p = \emptyset$ for all $m \neq m'$.
Given the partition, a set of associated links between users and providers is defined by $L_t = \cup_{m=1}^M \{(i,j)| \forall i \in G_m^u, \exists j \in G_m^p\} \cup \{(i,0)|\forall i \in G_0\}$.
Figure~\ref{fig:example_groupstructure}(a) illustrates an example of the group structure formed by five users and four providers where $G_0=\{4\}, G_1=\{\{1,2,3\},\{1,2,4\}\}, G_2=\{\{5\},\{3\}\}$, and $L=\{(1,2),(2,1),(3,2),(4,0),(5,3)\}$.
\begin{figure}[t]
  \centering
	\resizebox{1.0\linewidth}{!}{\includegraphics{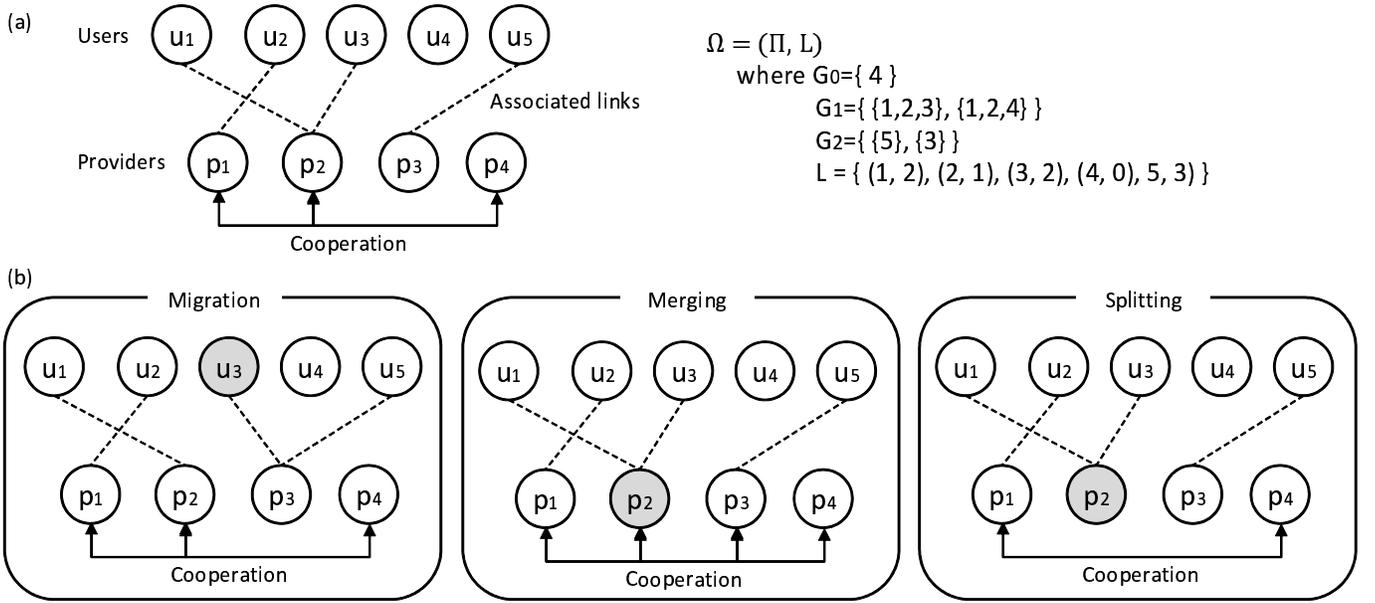}}
  \caption{(a) An example of group structure formed by 5 users and 4 providers. For the readability, we respectively label 'u' and 'p' for users and providers. (b) Example group structures after migration (by u$_3$), merging (by p$_2$), or splitting decision (by p$_2$).}
  \label{fig:example_groupstructure}
\end{figure}

\subsection{Instance Allocation}
\label{subsec:instanceallocation}

\subsubsection{Formulation of Social Welfare Optimization}
\label{subsec:optimization_formulation}
Users and providers are allowed to trade their instances in their own groups.
Given bids and offers in each group, the instance allocation of the group is formulated as an optimization problem to maximize the social welfare, i.e., the total utilities of the users and providers.

The utility of user $i$ in group $G$ with the allocation $A_t$ is given by
\begin{align}
\label{eqn:utility_user}
u_i(G, A_t)
&= v_i - c_i \nonumber \\
&= v_i - \sum_{s=t_i^s+1}^{t_i^e}  \sum_{j \in G^p} a_{ijs} \cdot p_{ijs}  \nonumber \\
&= \sum_{s=t_i^s+1}^{t_i^e} \sum_{j \in G^p} a_{ijs} \cdot (v_{is} - p_{ijs} ) 
\end{align}
where $c_i$ is a final charge to user $i$, $v_{is} \equiv \frac{v_i}{\ell_i}$ is user $i$'s valuation to instances $\vec{d_i}$ at time slot $s$, and $p_{ijs}$ is a trading price for the instances $\vec{d_i}$ between user $i$ and provider $j$ at time slot $s$.

The utility of provider $j$ in group $G$ with the allocation $A_t$ is given by
\begin{align}
\label{eqn:utility_prov}
u_j(G, A_t)
&= r_j - v_j  \nonumber \\
&= \sum_{i \in G^u} \sum_{s=t_i^s+1}^{t_i^e} a_{ijs} \cdot p_{ijs} - v_j \nonumber \\  
&= \sum_{i \in G^u} \sum_{s=t_i^s+1}^{t_i^e} a_{ijs} \cdot ( p_{ijs} - \sum_{k} d_i^k \cdot q_j^k[D_{js}^k] )
\end{align}
where $r_j$ is a provider $j$'s revenue and $v_j$ is a provider $j$'s valuation to the total amount of supplied instances $D_{js}^k = \sum_{i \in G^u} a_{ijs} \cdot d_i^k$ for $k=1,\ldots,K$.
%
Therefore, the instance allocation $A=[a_{ijs}]$ in group $G$ is formulated as follows:
\begin{align}
\label{eqn:optimization}
\max_{A_t} & \ \ \
\mathcal U(G, A_t) = \sum_{i \in G^u} u_i(G, A_t) + \sum_{j \in G^p} u_j(G, A_t)  \nonumber \\
&= \sum_{i \in G^u} \sum_{s=t_i^s+1}^{t_i^e} \sum_{j \in G^p} a_{ijs} \cdot (v_{is} - p_{ijs} ) 
   + \sum_{j \in G^p} \sum_{i \in G^u} \sum_{s=t_i^s+1}^{t_i^e} a_{ijs} \cdot ( p_{ijs} - \sum_{k} d_i^k \cdot q_j^k[D_{js}^k] ) \nonumber \\
&= \sum_{i \in G^u} \sum_{s=t_i^s+1}^{t_i^e} \sum_{j \in G^p} a_{ijs} \cdot ( \frac{v_i}{\ell_i} - \sum_{k} d_i^k \cdot q_j^k[D_{js}^k] ) \nonumber \\
&= \sum_{i \in G^u} x_i \cdot v_i - \sum_{i \in G^u} \sum_{s=t_i^s+1}^{t_i^e} \sum_{j \in G^p}  a_{ijs} \sum_{k} d_i^k \cdot q_j^k[D_{js}^k]  
\end{align}
subject to
\begin{align}
D_{js}^k = \sum_{i \in G^u} d_i^k \leq s_{js}^k,  \\
x_i \cdot \ell_i = \sum_{s=t_i^s+1}^{t_i^e} y_{is}, \\
y_{is} = \sum_{j \in G^p} a_{ijs}  \\
x_i \in \{0,1\}, \ y_{is} \in \{0,1\}, \ a_{ijs} \in \{0,1\} 
\end{align}
for $\forall i \in G^u$, $\forall j \in G^p$, $\underline{t} < s \leq \overline{t}$, $1 \leq k \leq K$
where $\underline{t} = \min_{} \{t_i^s | \forall i \in G^u\}$, $\overline{t}  = \max_{} \{t_i^e | \forall i \in G^u\}$.

Thus, the social welfare optimization problem is formulated to find the instance allocation that maximizes the total utilities of users and providers (Equation (3)) under constraints (4)-(6) of the amount of supplied instances and about the allocation.
Constraint (4) guarantees that the total number of requested/allocated instances does not exceed the total number of supplied instances.
Constraint (5) indicates that a user's demand is satisfied when all of the requested instances are allocated, which is specified by a binary variable $x_i$.
Constraint (6) indicates that user's instances are allocated to only one provider at each time slot, which is specified by a binary variable $y_{is}$.
Thus, the system does not consider the partially allocated instances.

\subsubsection{A Dynamic Algorithm to Find the Nearly Optimal Solution}
\label{subsec:optimization_algorithm}

The formulated binary integer programming is known to be NP-hard that the number of different allocations to be evaluated exponentially increases as the numbers of users and providers increase.
We propose a dynamic algorithm to find the nearly optimal solution in a reasonable time.

Algorithm~\ref{algInstanceAllocation} shows the dynamic algorithm to find the nearly optimal solution given bids $B_t$ and $O_t$ at time $t$.
The key idea of this algorithm is to give higher allocation priority to users who bid with higher valuation and providers who offer with lower valuation.
Besides, the algorithm allocates instances from the time slot that more users request their instances to the time slots that less users request because providers are likely to offer a discount price in such time slot.

At the beginning (Lines 1-9), the algorithm checks the users' demand periods and counts how many users request their instances over different time slots from $\underline{t}+1$ to $\overline{t}$.
$C$ is a count vector, whose element $C_s$ indicates the number of users whose demand periods include the time slot $s$.
Pos is a mapping between users and providers, whose element Pos$_i$ is a provider id that a user $i$ is assigned to. Pos$_i=0$ means that a user $i$ is not mapped to any providers.

In the following loop (Lines 12-37), the algorithm finds the allocation $A_t$ by starting to consider time slot $s$ in descending order of the counts $C$ (Line 12).
The first subloop (Lines 16-30) computes the best mapping between users and providers at time slot $s$.
According to the total number of instances requested by users in the same group (Line 21),
the users with relatively higher valuations are preferably assigned (Line 30) to providers with relatively lower valuations if such providers $j^*$ exist (Line 23).
If there is no such providers, then users consider other providers in different groups (Line 26).
If users cannot find any providers, then their demands cannot be assigned at this time slot (Line 25).
The second subloop (Lines 32-36) updates the recent allocation based on the computed mapping.
Once all instances requested by a user are mapped, the user $i$ and his bid $b_i$ are excluded (Line 36).

\begin{algorithm}[p]
\DontPrintSemicolon
\SetKwData{Pos}{Pos}
\SetKwData{Update}{update}
\SetKwData{GM}{$\overline{g}$}
\KwIn{$B_t$: a set of bids (from a set of users $I_t$) \\
\Indp\Indp\Indp
$O_t$: a set of offers (from a set of providers $J_t$) \\
}
\KwOut{$A_t=[a_{ijs}]$: an allocation \\
}
$A_t \leftarrow 0$;\\
\mbox{{\it /* Count the number of users whose demand periods are in timeslots $s$ */}} \\
$\underline{t} \leftarrow \min \{t^s_i \ | \ \forall b_i \in B_t\}$\\ 
$\overline{t} \leftarrow \max \{t^e_i \ | \ \forall b_i \in B_t\}$\\
\For{timeslot $s \leftarrow \underline{t}+1$ to $\overline{t}$} {
$C_s \leftarrow 0$ \\
\ForEach{user $i \in I_t$} {
\lIf{$t^s_i < s \leq t^e_i$} {$C_s \leftarrow C_s + 1$}
} 
} 
$C \leftarrow (C_{\underline{t}+1}, C_{\underline{t}+2}, \ldots, C_{\overline{t}})$ \hfill\mbox{{\it // A count vector}} \\
$\Pos \leftarrow (\Pos_1,\ldots,\Pos_{|I_t|})$ \hfill\mbox{{\it // A mapping between users and providers }} \\
\mbox{{\it /* Find the allocation from timeslot in descending order of $C_s$ */}} \\
\ForEach{timeslot $s$ in $sort(C, $'descend'$)$} {
$G_g \leftarrow \emptyset$ for $g=1,\ldots,|J_t|$  \\
\mbox{{\it /* Handling users in descending order of valuation $\frac{v_i}{\ell_i}$ */}} \\
$\mathcal U \leftarrow \emptyset$ \hfill\mbox{{\it // A set of allocated users }} \\
\ForEach{user $i \in I_t$ such that $t_i^s < s \leq t_i^e$} {
$\GM \leftarrow \min \{g \ | \ G_g = \emptyset, g=1,\ldots,|J_t| \}$  \\
$\mathcal P \leftarrow \emptyset$ \hfill\mbox{{\it // A set of allocated providers }} \\
\For{$g \leftarrow 1 \ to \ \GM$} {
\mbox{{\it /* $\mu = 1$ if $i \notin \mathcal U$, and $0$ otherwise */}} \\
$\vec{D_s} \leftarrow \mu \cdot \vec{d_i} + \sum_{m \in G_g} \vec{d_m}$ \\
\mbox{{\it /* Find a provider with the lowest valuation */}} \\
$j^* \leftarrow \arg\min_{j \in J_t} \{ v_j = \sum_{k} D_s^k \cdot q_j^k[D_s^k]$ $|$ $v_j \leq v_i$, $\vec{D_s} \leq \vec{s_{js}}, j \notin \mathcal P  \}$ \\
\uIf{$j^*==0$} {
\lIf{$g==\GM$}{$\Pos_i \leftarrow 0$} \\
\lElse{$\mathcal P \leftarrow \mathcal P \cup \{ \Pos_m$ $|$ $\exists m \in G_g\}$} \\
}
\Else{
$\mathcal U \leftarrow \mathcal U \cup \{i\}$ \\
$G_g \leftarrow G_g \cup \{i\}$     \\
$\Pos_m \leftarrow j^*$ for $\forall m \in G_g$ \\
}
} 
} 
\mbox{{\it /* Update the allocation information */}} \\
\ForEach{user $i \in I_t$} {
\If{$\Pos_i \neq 0$}{
$a_{i,\Pos_i,s} \leftarrow 1$; \ 
$\ell_i \leftarrow \ell_i-1$ \\
}
\If{$\ell_i==0$} {
$I_t \leftarrow I_t - \{i\}$; \ 
$B_t \leftarrow B_t - \{b_i\}$ \\
} 
} 
} 
\Return{$A_t = [a_{ijs}]$}
\caption{findInstanceAllocation($B_t$, $O_t$)}
\label{algInstanceAllocation}
\end{algorithm}


\subsubsection{Trading Price Determination}
\label{subsubsec:pricing}

Based on the allocation, the trading price is determined in a way that the social welfare is distributed to participants in proportional to their contributions, i.e., valuations.
First, the controller computes the total price charged to users for each group $G$.
The price at time slot $s$ is given as follows:
\begin{align}
\label{eqn:price}
p_{Gs}
&= \kappa \cdot \sum_{i \in G^u} v_{is} + (1-\kappa) \cdot  \sum_{j \in G^p} v_{js} \\
&= \kappa \cdot \sum_{i \in G^u} \frac{v_i}{\ell_i} + (1-\kappa) \cdot \sum_{j \in G^p} \sum_{k} D_{js}^k \cdot q_j^k[D_{js}^k]
\end{align}
where $\kappa \in [0,1]$.
$\kappa$ is set to 0.5 in this paper to give users and providers the fairness.

Second, the users in $G$ divide the total price $p_{Gs}$ in weights proportional to the user's valuation.
Hence, the price for user $i$ is defined by
\begin{align}
\label{eqn:tradingprice}
p_{ijs}
&= \frac{v_{is}}{\sum_{i \in G^u} v_{is}} \cdot p_{Gs}  \\
&= \frac{v_{is}}{V_{s}} \cdot (\kappa \cdot \sum_{i \in G^u} v_{is} + (1-\kappa) \cdot \sum_{j \in G^p} v_{js}) \\
&= a_{ijs} \cdot (\kappa \cdot v_{is} + (1-\kappa) \cdot \frac{v_{is}}{V_{s}} \cdot v_{js})		\label{eqn:pricing}
\end{align}
where $V_{s} = \sum_{i \in G^u} v_{is}$ is the total valuation of users in $G$.

Consequently, the cost (i.e., the total charge) for user $i$ is
\begin{align}
\label{eqn:cost}
c_{i}
= \sum_{s=t_i^s+1}^{t_i^e} \sum_{j \in J_t} p_{ijs}
, 
\end{align}
and the revenue for provider $j$ is
\begin{align}
\label{eqn:profit}
r_{j}
= \sum_{i \in I_t} \sum_{s=t_i^s+1}^{t_i^e} p_{ijs}
.
\end{align}


\subsection{A Group Formation Algorithm}
\label{subsec:groupformationgame}

In the proposed system, the social welfare is further improved by the group formation algorithm that leverages the concept of cooperation among users and providers (as described in Section~\ref{subsubsec:cooperation}).
The proposed group formation algorithm allows participants to make {\it group formation decisions} to increase their own payoffs.

\subsubsection{Payoff}
\label{subsubsec:payoff}

The payoff for user $i$ is designed as the difference between the utility and the delaying penalty by forming a group and defined as follows:
\begin{align}
\label{eqn:payoff_user}
\phi_i^u(G) = u_i(G, A_t^*) - \xi_i(G, A_t^*)
\end{align}
where $A_t^* = \arg \max_{A_t} \mathcal U(G, A_t)$, and the delaying penalty is defined by
\begin{align}
\label{eqn:penalty}
\xi_i (G, A_t^*)
&= \sum_{s=t_i^s+1}^{t_i^e} a_{ijs}^* \cdot z_{is} \cdot C^d \\
\mbox{where }
&  z_{is}
   = \left\{ 
     \begin{array}{l l}
         s - t_i^s - \ell_i & \quad \text{if $s > t_i^s + \ell_i$} \\
         0 & \quad \text{otherwise.}\\
     \end{array} \right. 
\end{align}
The actual cost $C^d$ may differ according to different participants, applications, and time slots. However, for simplicity, we assume that the actual cost is proportional to the delay.

Similarly, the payoff for provider $j$ is designed as the difference between the utility and the additional cost by migrating instances and defined as follows:
\begin{align}
\label{eqn:payoff_provider}
\phi_j^p(G) = u_j(G, A_t^*) - \xi_j(G, A_t^*)
\end{align}
where the cost for instance migration is defined by
\begin{align}
\label{eqn:utility}
\xi_j (G, A_t^*)
&= \sum_{s=\underline{t}+1}^{\overline{t}} m_{js} \cdot C^m \\
\mbox{where }
&  m_{js}
   = \left\{ 
     \begin{array}{l l}
         1 & \quad \text{if $a_{ijs}^*=1 \wedge a_{ij's-1}^*=1$, $j \neq j'$} \\
         0 & \quad \text{otherwise.}\\
     \end{array} \right.
\end{align}
The migration cost $C^m$ also differs among providers. However, we assume the value to be a constant for all providers.

\subsubsection{The Proposed Group Formation Algorithm}
\label{subsubsec:groupformationalgorithm}

For constructing a group formation, we introduce an algorithm that allows users and providers to make the group formation decisions for selecting which groups (i.e., providers) to join and which users to accept.
The decisions are called {\it migrating}, {\it merging}, and {\it splitting}.

{\it Migrating}:
Given a group structure $\Omega=(\Pi=(G_1, \ldots, G_m, \ldots, G_M), L)$,
any user $i \in G_m^u$ decides to migrate from its current associated provider $j \in G_m^p$ to another provider $j' \in G_{m'}^p$ where $m \neq m'$
if the following condition is satisfied:
\begin{align}
\label{eqn:mig1}
\phi^u_i(G_{m'} \cup \{i\}) & > \phi^u_i(G_m)  
.
\end{align}
The condition (\ref{eqn:mig1}) guarantees for a user to improve its own payoff by performing the migrating decision.
For every single migration of user $i$, the group structure $\Omega$ is updated to
$\Omega'=(\Pi', L')$ where
$\Pi'= (\Pi \backslash \{G_m, G_{m'}\}) \cup \{ G_m \backslash \{i\}, G_{m'} \cup \{i\} \}$, and
$L'= L \backslash (i,j) \cup (i,j')$.

{\it Merging}:
Multiple groups $G_m \in \Pi_s$, a subset of $\Pi$ (i.e., $\Pi_s \subseteq \Pi$) can collectively form a single group $G_{m'}^{\dagger}$
if the following conditions are satisfied:
\begin{align}
\label{eqn:mer1}
\mbox{$\exists j \in G_{m}^p$}, \phi^u_{j}(G_{m'}^\dagger) & > \phi^u_{j}(G_{m}) 
    \mbox{, $G_m \in \Pi_s$} \\
\label{eqn:mer2}
\phi^u_{i}(G_{m'}^\dagger) & \geq \phi^u_{i}(G_{m}) 
    \mbox{, $\forall i \in G_{m}^u$, $\forall G_m \in \Pi_s$} \\
\label{eqn:mer3}
\phi^p_{j'}(G_{m'}^\dagger) & \geq \phi^p_{j'}(G_{m}) 
    \mbox{, $\forall j' \in G_{m}^p$, $j' \neq j$, $\forall G_m \in \Pi_s$} 
\end{align}
where $G_{m'}^\dagger = \cup_{\Pi_s} G_m$.
The first condition in (\ref{eqn:mer1}) guarantees for a provider to improve its own payoff by performing the merging decision. The other conditions guarantee that none of the other members (both users and providers) in the newly formed group decreases their payoffs by the merging decision.
For every single merging decision, the group structure $\Omega$ is updated to
$\Omega^{\dagger}=(\Pi^{\dagger}, L)$ where
$\Pi^{\dagger}= (\Pi \backslash \{\forall G_m \in \Pi_s\}) \cup G_{m'}^{\dagger}$ where $|\Pi^\dagger| < |\Pi|$.

{\it Splitting}:
Given the original group $G_m$, members in this group can collectively split into multiple groups $G_{m'}^{\ddagger}$, whose set after the splitting is denoted by $\Pi_s^{\ddagger}$,
if the following conditions are satisfied:
\begin{align}
\label{eqn:spl1}
\exists j \in G_{m}^p, \exists G_{m'}^\ddagger \in \Pi_s^\ddagger, \phi^p_{j}(G_{m'}^\ddagger) & > \phi^p_{j}(G_{m}) 
    \mbox{, $G_m \in \Pi_s$} \\
\label{eqn:spl2}
\exists G_{m'}^\ddagger \in \Pi_s^\ddagger, \phi^u_{i}(G_{m'}^\ddagger) & \geq \phi^u_{i}(G_{m}) 
    \mbox{, $\forall i \in G_{m}^{u}$} \\
\label{eqn:spl3}
\exists G_{m'}^\ddagger \in \Pi_s^\ddagger, \phi^p_{j'}(G_{m'}^\ddagger) & \geq \phi^p_{j'}(G_{m}) 
    \mbox{, $\forall j' \in G_{m}^{p}$, $j' \neq j$}
\end{align}
where $G_{m} = \cup_{\Pi_s^\ddagger} G_{m'}^\ddagger$ and
for all $G_{m'}^\ddagger$ there exist $(i,j) \in L$ such that $i \in G_{m'}^{\ddagger u}$ and $j \in G_{m'}^{\ddagger p}$.
The first condition in (\ref{eqn:spl1}) guarantees for a provider to improve its own payoff by performing the splitting decision. 
The other conditions guarantee that none of the members in newly formed groups decreases their payoffs by the splitting decision.
For every single splitting decision, the group structure $\Omega$ is updated to
$\Omega^{\ddagger}=(\Pi^{\ddagger}, L)$ where
$\Pi^{\ddagger}= (\Pi \backslash G_m) \cup \{ \forall G_{m'}^{\ddagger} \in \Pi_s^\ddagger  \}$ where $|\Pi^\ddagger| > |\Pi|$.


Figure~\ref{fig:example_groupstructure}(b) illustrates example group structures after migrating, merging, or splitting decision given the group structure in Figure~\ref{fig:example_groupstructure}(a).
For example, after the migration by u$_3$, the group structure changes to $G_0=\{4\}, G_1=\{\{1,2\},\{1,2,4\}\}, G_2=\{\{3,5\},\{3\}\}$, and $L=\{(1,2),(2,1),(3,3),(4,0),(5,3)\}$.
After the merging by p$_2$, the groups 1 and 2 become one group $G_1=\{\{1,2,3,5\},\{1,2,3,4\}\}$ with the same associated links.
After the splitting by p$_2$, the group 1 is splitted into $G_1=\{\{2\},\{1,4\}\}, G_3=\{\{1,3\},\{2\}\}$ with the same associated links.

Algorithm~\ref{algGroupFormation} shows a pseudocode of the proposed group formation algorithm that finds a stable group structure.
At the beginning (Line 3), the algorithm initializes a group structure $\Omega$ by randomly setting groups and associated links (Lines 4 and 5 in Algorithm~\ref{algGroupFormation-init}). 
A randomly initialized group structure converges to the stable group structure for participants by repeatedly selecting distributed decisions (Lines 9, 13, and 16) that improve both users and providers' payoffs.
During the iterations (Lines 4-18), a history set $H$ is maintained, which stores the groups that are formed or evaluated in the past.
To compute the payoffs, Algorithm 1 is executed. 
The algorithm stops when none of the participants changes their groups (Line 18).

\begin{algorithm}[t]
\DontPrintSemicolon
\SetKwData{Itr}{itr}
\KwIn{$B_t$: a set of bids (from a set of users $I_t$) \\
\Indp\Indp\Indp
$O_t$: a set of offers (from a set of providers $J_t$) \\
$\Omega_t$: the current group structure at time $t$ \\
}
\KwOut{$\Omega_t^F$: the final group structure \\
}
$H \leftarrow \emptyset$ \hfill\mbox{{\it // A history set of previously seen groups}} \\
\mbox{{\it /* Randomly initialize a group structure at the beginning */}} \\
\lIf{$\Omega_t$ == $\emptyset$}{
$\Omega \leftarrow initGroupStructure(I_t, J_t)$ \hfill\mbox{{\it // Algorithm 3}} \\
}
\Repeat{None of users and providers changes their groups.} {
\mbox{{\it /* Given the current group structure $\Omega$ */}} \\
\ForEach{ user $i \in I_t$} {
\mbox{{\it /* $i$ searches a provider $j'$ that satisfies (\ref{eqn:mig1}).}} \\
\mbox{{\it \ \ \ If such $j'$ is found, and $G_{m'} \cup {i} \notin H$, then $i$ performs a migrating operation. */}} \\
$\Omega_t \leftarrow Migrating(i, \Omega)$ \\
}
\ForEach{ provider $j \in J_t$} {
\mbox{{\it /* $j$ searches a possible group $G_{m'}$ that satisfies (\ref{eqn:mer1})-(\ref{eqn:mer2}).}} \\
\mbox{{\it \ \ \ If such group is found, and $G_{m'} \notin H$, then $j$ performs a merging operation. */}} \\
$\Omega_t \leftarrow Merging(\Omega, j)$ \\
\mbox{{\it /* $j$ searches a possible partition of its own group $G_m$, which satisfy (\ref{eqn:spl1})-(\ref{eqn:spl2}).}} \\
\mbox{{\it \ \ \ If such partition is found, and $G_{m'}^\ddagger \notin H$, then $j$ performs a splitting operation. */}} \\
$\Omega_t \leftarrow Splitting(\Omega, j)$ \\
}
$H \leftarrow H \cup G, \forall G \in \Pi$ of $\Omega_t$
}
\Return{ $\Omega^F_t \leftarrow \Omega_t$}
\caption{findGroupStructure($B_t$, $O_t$, $\Omega_t$)}
\label{algGroupFormation}
\end{algorithm}

\begin{algorithm}[t]
\DontPrintSemicolon
\SetKwData{Itr}{itr}
\KwIn{$I_t$: a set of users, $J_t$: a set of offers  \\
}
\KwOut{$\Omega^0$: an initial group structure \\
}
$G_j^u, G_j^p \leftarrow \emptyset$ \mbox{, for $j=1,\ldots,|J_t|$} \\
\ForEach{ user $i \in I_t$} {
$m \leftarrow rand(1,|J_t|)$ \\
$G_{m}^u \leftarrow G_{m}^u \cup \{i\}$ \\
$L \leftarrow L \cup (i,m)$ \\
}
\ForEach{ provider $j \in J_t$} {
$G_j \leftarrow G_j^u \cup (G_j^p = \{j\})$ \\
}
Return{ $\Omega^0 = \{G_1, \dots, G_{|J_t|}\}$}
\caption{initGroupStructure($I_t$, $J_t$)}
\label{algGroupFormation-init}
\end{algorithm}


\section{Analysis of the Group Formation Algorithm}
\label{sec:analysis}

This section gives a theoretical analysis of the proposed group formation algorithm in terms of the time complexity and stability.
Subsequently, we discuss that the proposed group auction system holds three auction properties described in Section~\ref{subsec:overview}.

\subsection{Time complexity}

The binary integer programming formulated in Section~\ref{subsec:optimization_algorithm} is known to be NP-hard that the number of different allocations to be evaluated exponentially increases as the numbers of users and providers increase.
For example, the numbers of users and providers are $N$ and $M$, respectively. There exist $T(N,M)$ possible mappings (i.e., instance allocations) between the users and providers where $T(n,m)=\sum_{i=0}^{n} \binom{n}{i} \cdot T(n-i,m-1)$ when $T(0,k) \equiv 1$ for $\forall k \in \mathbb N^0$ and $T(k^+,0) \equiv 0$ for $\forall k^+ \in \mathbb N^+$.\footnote{$\mathbb N^0$ indicates a set of non-negative integers, and $\mathbb N^+$ is a set of positive integers.}
In addition, if we consider the number of time slots $S$ in a demand period, then the number of possible combinations of the mappings increases to $T(N,M)^S$, which is factorial $O(N!)$ in O-notation.

The proposed group formation by Algorithm~\ref{algGroupFormation} is quadratic to the number of participants.
The algorithm contains a loop (Lines 4-18) that evaluates decisions of participants. For each single decision, we need to compute the instance allocation by Algorithm ~\ref{algInstanceAllocation} whose complexity is linear $O(S$$\cdot$$N$$\cdot$$M)$ to participants and time slots.
If the loop is executed for $k$ times, then Algorithm 1 is run for $k \cdot (N$+$M)$ times.
Therefor, the complexity of the proposed group formation is $O(k$$\cdot$$(N$+$M)$$\cdot$$S$$\cdot$$N$$\cdot$$M)$ $\sim$ $O(k$$\cdot$$S$$\cdot$$N^2)$ in O-notation when $N \gg M$.

\subsection{Stability}

The stability of the proposed group formation algorithm is guaranteed by the concept of Nash equilibrium described in \cite{Bogomonlaia02}.
We relax the stability conditions by considering a relaxing parameter $\epsilon \geq 0$ (epsilon) as follows.


\newtheorem{df}{Definition}
\newtheorem{thm}{Theorem}
\newtheorem{pro}{Proposition}

\begin{thm}
\label{thm1}
Given any initial group structure, the proposed group formation algorithm always terminates, i.e., converges to a final group structure.
\end{thm}
\begin{proof}
Let $\Omega^{r}_{n_r}$ denote the group structure formed after $n_r$ decisions made by one or more participants during $r$ iterations of Algorithm~\ref{algGroupFormation} (Lines 4-18).
Starting from any initial group structure $\Omega^0_0$, the proposed algorithm iteratively transforms the group structure into another group structure by taking a sequence of migrating, merging, and/or splitting decisions.
The sequence of decisions yields the following transformations of the group structure
\begin{align}
\label{transformation}
\Omega^{0}_{0} = \Omega^{1}_{0} \rightarrow \Omega^{1}_{1} \rightarrow \cdots \rightarrow \Omega^{1}_{n_1} = \Omega^{2}_{n_1} \rightarrow \cdots \rightarrow \Omega^{r}_{n_r} = \Omega^{r+1}_{n_r} \rightarrow \cdots \rightarrow \Omega^{T}_{n_T}
\end{align}
where the operator $\rightarrow$ indicates the occurrence of one of three decisions.
In other words, $\Omega^{r}_{n_r} \rightarrow \Omega^{r+1}_{n_{r+1}}$ implies that during the $(r+1) th$ iteration, $n_{r+1} - n_r$ decisions are made by participants, which results in a new group structure $\Omega^{r+1}_{n_{r+1}}$ starting from a group structure $\Omega^{r}_{n_{r}}$.
By inspecting the conditions for the decision making defined in (\ref{eqn:mig1})-(\ref{eqn:spl2}) and a history set of previously seen groups, it can be seen that every single decision leads to a new group structure that has not yet visited.
Hence, for any two group structures $\Omega^{a}_{n_{a}}$ and $\Omega^{b}_{n_{b}}$ where $n_a \neq n_b$ in the transformations in (\ref{transformation}), it is true that $\Omega^{a}_{n_{a}} \neq \Omega^{b}_{n_{b}}$.

Given this property and the well known fact that the number of possible group structures formed by a finite set of participants (i.e., $I$ and $J$) is finite and given by the Bell number \cite{DRay07}, the number of transformations is finite.
Therefore, the proposed group formation algorithm always terminates and converges to one group structure $\Omega^T_{n_T} \equiv \Omega^F$, a final group structure.
\end{proof}

\begin{df}[$\epsilon$-Nash equilibrium]
A group structure $\Omega = (\Pi, L)$ is $\epsilon$-{\it Nash-stable} if for all users $i \in I$, $\phi_i^u(G_m) \geq \phi_i^u(G_{m'} \cup \{i\}) - \epsilon$, $i \in G_m^u$, and for all providers $j \in J$, $\phi_j^p(G_m) \geq \phi_j^p(G_{m'} \cup \{j\}) - \epsilon$, $j \in G_m^p$, for all other groups $G_{m'} \in \Pi$, $G_m \neq G_{m'}$.
Nash equilibrium is a special case of $\epsilon$-Nash equilibrium with $\epsilon=0$.
\end{df}
In the $\epsilon$-Nash-stable group structure, no users/providers can gain an extra payoff more than $\epsilon$ by  unilaterally moving from its current group to another group.

\begin{pro}
\label{pro:stability}
Any final group structure $\Omega^F$ obtained by the proposed group formation algorithm is $\epsilon^*$-Nash-stable.
\end{pro}
\begin{proof}
Assume that the group structure $\Omega^{F}=(\Pi^{F}, L^{F})$ obtained by the proposed algorithm is not Nash-stable (i.e., $\epsilon=0$). There are two cases to be considered.

Case 1:
There exists a user $i$ in $G_m$ and a group $G_{m'} \in \Pi^{F}$ such that $\phi_i^u(G_{m'} \cup \{i\}) > \phi_i^u(G_m)$.
Therefore, the user can perform a migrating decision to improve its own payoff.
Similarly, there exists a provider $j$ in $G_m$ and a group $G_{m'} \in \Pi^{F}$ such that $\phi_j^p(G_{m'} \cup \{j\}) > \phi_j^p(G_m)$.
This implies that the provider can perform merging and splitting decisions to improve its own payoff under the true conditions in (\ref{eqn:mer2}) and (\ref{eqn:mer3}), and (\ref{eqn:spl2}) and (\ref{eqn:spl3}).
This contradicts with the fact that the group structure $\Omega^{F}$ obtained by the proposed algorithm is the converged group structure, proved in Theorem~\ref{thm1}.
 
Case 2: 
Although the provider has an incentive to move from its current group to another group, no further transformation by performing a merging (or splitting) decision can be made when the condition in (\ref{eqn:mer2}) or (\ref{eqn:mer3}) (or (\ref{eqn:spl2}) or (\ref{eqn:spl3})) is false.
If we relax the stability conditions by considering a relaxing parameter $\epsilon^* > 0$, then $\Omega^F$ is $\epsilon^*$-Nash-stable where
$\epsilon^* \equiv (\epsilon^{*u},\epsilon^{*p})$ is defined by
\begin{align*}
\epsilon^{*u} &= \max_{\forall i \in I, \forall G_{m'} \in \Pi^F} \Big\{ \max \Big(0, \phi_i^u(G_m' \cup \{i\}) - \phi_i^u(G_m) \Big) \Big\} \mbox{, and} \\
\epsilon^{*p} &= \max_{\forall j \in J, \forall G_{m'} \in \Pi^F} \Big\{ \max \Big(0, \phi_j^p(G_m' \cup \{j\}) - \phi_j^p(G_m) \Big) \Big\}.
\end{align*}
\end{proof}

As we can easily see, the lower value of $\epsilon^*$ is preferred in terms of the perfection of the Nash stability.

\subsection{Auction Properties}

This section discusses about three auction properties (allocative efficiency, individual rationality, and budget balance)
presented in Section~\ref{subsec:overview} and shows that the proposed group auction theoretically holds the properties.

\begin{pro}
The instance allocation associated with the group structure obtained by the proposed group formation algorithm is allocatively efficient.
\end{pro}
\begin{proof}
Given a group structure and the associated instance allocation, assume that there exists a provider who can gain extra utility by merging existing groups or splitting its current group.
However, the conditions in (\ref{eqn:mer2}) and (\ref{eqn:mer3}), or (\ref{eqn:spl2}) and (\ref{eqn:spl3}) for the decisions restrict the merging or splitting which reduces the others' utilities. Thus, the provider does not perform the decision.
For users, migrating to the other group does not decrease the other users' utilities due to the assumption of the non-increasing price curve. That is, a trading price of instances does not increase when the number of users increases.
Consequently, no providers gain utility by decreasing the others' utilities.
\end{proof}

\begin{pro}
The instance allocation associated with the group structure obtained by the proposed group formation algorithm is individually rational.
\end{pro}
\begin{proof}
This proposition is guaranteed by the line 22 in Algorithm~\ref{algInstanceAllocation} and the trading price computed from the expression in (\ref{eqn:pricing}).
Any user's valuation $v_i$ does not exceed any provider's valuation $v_j$ derived from its price curve and the allocation.
In addition, the trading price for instances is ranged between the $v_i$ and $v_j$ due to a parameter $\kappa \in [0,1]$.
\end{proof}

\begin{pro}
The instance allocation associated with the group structure obtained by the proposed group formation algorithm is budget-balanced.
\end{pro}
\begin{proof}
This proposition is guaranteed by our design of trading price described in Section~\ref{subsubsec:pricing}.
From Equations in (\ref{eqn:cost}) and (\ref{eqn:profit}), 
\begin{align}
\sum_{i \in I_t} c_i
 = \sum_{i \in I_t} \sum_{s=t_i^s+1}^{t_i^e} \sum_{j \in J_t} p_{ijs}
 = \sum_{j \in J_t} \sum_{i \in I_t} \sum_{s=t_i^s+1}^{t_i^e} p_{ijs}
 = \sum_{j \in J_t} r_j
,
\end{align}
which implies that the total payments by users is the same as the total profits of providers.
\end{proof}

\section{Evaluation}
\label{sec:evaluation}

This section numerically evaluates the convergence  (in Section \ref{subsec:eval_convergence}) and optimality  (in Section \ref{subsec:eval_optimality}) of the proposed group formation algorithm and presents the advantages of the proposed real-time group auction system by showing comparison results with different schemes through the simulation (in Section \ref{subsec:simulationresults}).
 

\subsection{Convergence}
\label{subsec:eval_convergence}

We verify the convergence of the proposed group formation algorithm from two perspectives.
The first is whether the algorithm always converges over iterations given any parameter setting (e.g., the number of users, their demands, valuation, etc.)
The second is how the randomly given initial group structure (Line 2 in Algorithm~\ref{algGroupFormation}) impacts the final group structure (Line 18 in Algorithm~\ref{algGroupFormation}).

Figure~\ref{fig:sim_convergence_config} shows how social welfare, i.e., the total utilities defined by Equation (\ref{eqn:optimization}), improves over iterations at given instance allocation in different configurations of bids and offers.
In the figure, x-axis is the iteration index of the group formation algorithm, and y-axis indicates the total utilities of users and providers, whose value is normalized to the value at Iteration 20.
The result is obtained with 8 users and 2 providers requesting for/providing for one instance type under following bids and offers generated based on parameter setting described in Table~\ref{tab:sim_parametersetting1}.
For the user's valuation, we simply set the maximum valuation-per-unit as 10 cents (i.e., \$0.10 for one instance-per-time-slot), and then we multiply the unit price by the number of instances and length that a user requests.
For example, for a user who requests instances of $d_i=5$ for $\ell_i=3$, the value of his valuation is uniformly generated between 0 and $10\cdot5\cdot3=150$ cents (=\$1.50).
For the provider's valuation, price curves are randomly generated with $q_j[1] \in [5, 10]$ 
followed by $q_j[1] = \cdots = q_j[n-1] > q_j[n] = \cdots = q_j[s_j]$ where $n \in [2,s_j]$.
The line plots in the Figure~\ref{fig:sim_convergence_config} are the results of 10 randomly generated configurations of bids and offers.
The initial group structure is set to $\Pi=(G_0=\{1,\ldots,8\}$, $G_1=\{\emptyset, \{1\}\}$, $G_2=\{\emptyset, \{2\}\})$ and $L=\{(i,0)|i=1,\ldots,8\}$ for all cases.
The values are normalized to the values at Iteration 20.

In Figure~\ref{fig:sim_convergence_config}, we observed three trends in the value improvement.
For example, the value increases more at early iterations (Trend 1), the value increases more at later iterations (Trend 2), and the value gradually increases over some iterations (Trend 3). 
We observe that for all of 10 different configurations, Algorithm~\ref{algGroupFormation} converges to some values after several iterations.
The optimality will be discussed in the next subsection.

Figure~\ref{fig:sim_convergence_init} shows how the initial group structure impacts the final group structure.
For this measurement, we generate one configuration as described in Table~\ref{tab:sim_convergence1_example}.
We run Algorithm~\ref{algGroupFormation} ten times with 10 different randomly generated initial group structures.
The solid-line plots indicate the social welfare improvement over iterations.
A dashed line indicates the sum of the values $\epsilon^*$ (i.e., a relaxing parameter) of participants, which corresponds to one case denoted by a circle-pointed line.
It is seen that, during the process of group formation, the sum of $\epsilon^*$ gradually decreases and eventually reaches to 0.
This implies that the obtained group structure is Nash stable (i.e., $\epsilon^*$-Nash stable with $\epsilon^*$=0) in this case.
We verify that the algorithm finds one group structure whose social welfare is 118 regardless of initial group structures. 

\begin{table}[t]
  \caption{Parameter setting for bids and offers. This is used in Figures~\ref{fig:sim_convergence_config} and \ref{fig:sim_optimality}}
  \label{tab:sim_parametersetting1}
\centering
\small
	\begin{tabular}{|l|c||l|c|}
	\hline
	\multicolumn{2}{|c||}{User's bid} &
	\multicolumn{2}{|c|}{Provider's offer}
  \\ \hline
	Demand: \# of instances  &  $d_i \in [0, 10]$ 	&
	Supply: \# of instances  &  $s_j \in [10, 30]$
	\\
	Demand Period: Length & $\ell_i \in [1, 6]$ 	&
	Supply Period: Length & $w_j \in [1, 6]$
	\\
	\hspace{+24mm} Starting time  &  $t_i^s \in [1, 3]$ 	&
	\hspace{+22mm} Starting time  &  $t_j^s \in [1, 3]$
	\\
	\hspace{+24mm} Ending time & $t_i^e \in [t_i^s$+$\ell_i, \ t_i^s$+$\ell_i$+6] 	&
	\hspace{+22mm} Ending time & $t_j^e = t_j^s$+$w_j$
	\\
	Valuation & $v_i \in [0,V]$ 	&
	Valuation & $q_j[1] \in [5,10]$
	\\
													 & $V$=$10 \cdot d_i \cdot \ell_i$ 	&
													 & $q_j[n]\in [1,q_j[1]$-$1]$
	\\ \hline
	\end{tabular} 
\end{table}
\begin{table}[t]
  \caption{An example bids/offers configuration for 8 users and 2 providers, generated under parameter setting in Table I. This is used in Figures \ref{fig:sim_convergence_init} and \ref{fig:sim_convergence1_examplestructure}.}
  \label{tab:sim_convergence1_example}
\centering
\footnotesize
	\begin{tabular}{|l|c|c|c|c|c|c|c|c||l|c|c|}
	\hline
	\multicolumn{9}{|c||}{Bids} &
	\multicolumn{3}{|c|}{Offers}
  \\ \hline
  Demand  & $u_1$ & $u_2$ & $u_3$ & $u_4$ & $u_5$ & $u_6$ & $u_7$ & $u_8$ &
  Supply  & $p_1$ & $p_2$
  \\ \hline
	\# of instances & 2 &	2 &	2 &	5 &	5 &	5 &	10 & 10 &	
	\# of instances &	20 & 20	
	\\ \hline
	Period: Length & 4 & 5 & 6 & 4 & 5 & 6 & 4 & 6 &	
	Period: Length & 8 & 8
	\\ 
	\hspace{+9mm} Starting time  & 1 & 1 & 1 & 2 &	2 & 2 & 1 & 1 &	
	\hspace{+9mm} Starting time  & 1 & 1 
	\\
	\hspace{+9mm} Ending time  & 6 & 6 &	6 &	8 &	8 &	8 &	8 & 8 &	
	\hspace{+9mm} Ending time  & 8 & 8
	\\ \hline
	Valuation  & 8 & 10 & \bf 20 & \bf 10 & 20 & 30 & 40 & 60 &
	Valuation  & $q_1[1]$=.50 & $q_2[1]$=.60
	\\
	(per-unit valuation) & (1) & (1) & (1.7) & (0.5) & (0.8) & (1) & (1) & (1) &
	  & $q_1[15]$=.40 & $q_2[15]$=.30
	\\ \hline
	\end{tabular} 
\end{table}

\begin{figure}[t]
\centering
 \subfigure[\small For 10 different bids/offers configurations]{
   \includegraphics[scale = 0.63] {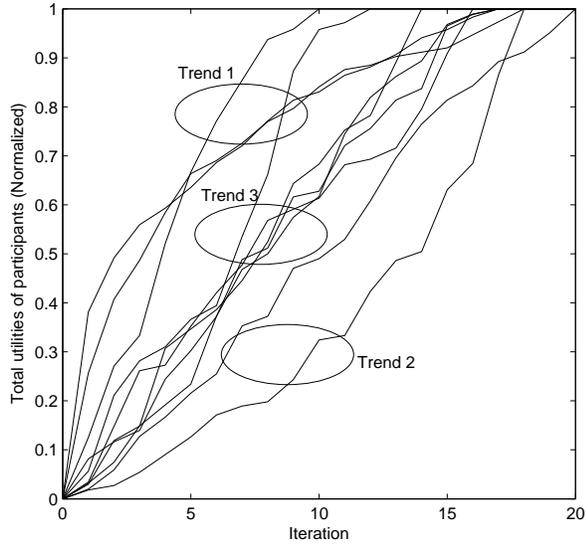}
  \label{fig:sim_convergence_config}
 }
 \subfigure[\small For 10 different initial group structures]{
   \includegraphics[scale = 0.63] {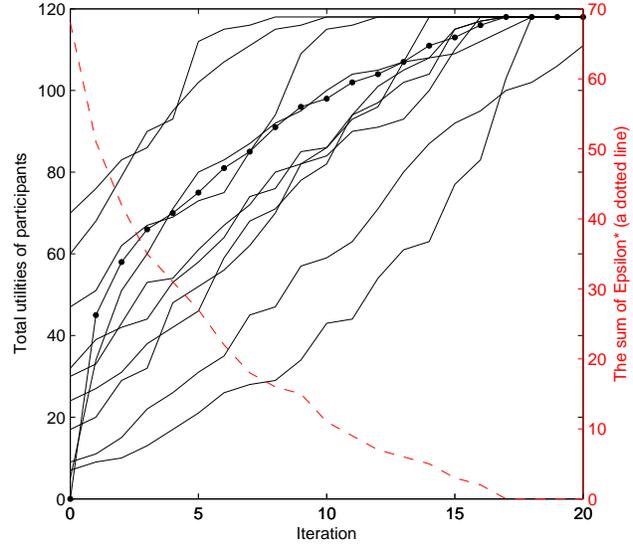}
  \label{fig:sim_convergence_init}
 }
 \subfigure[\small The obtained instance allocation]{
   \includegraphics[scale = 0.75] {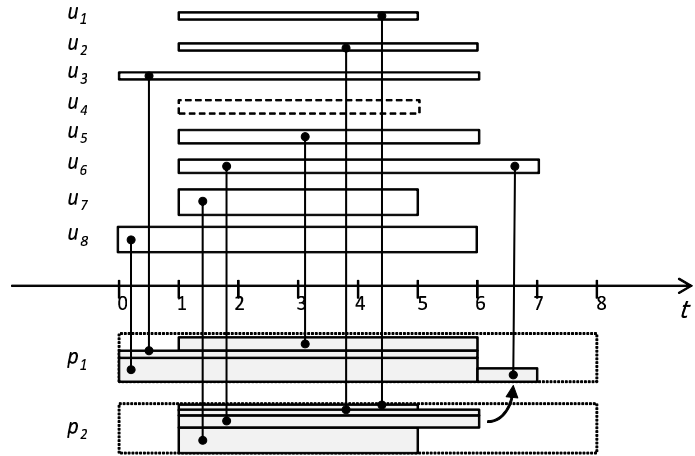}
  \label{fig:sim_convergence1_examplestructure}
 }
\caption{Impact of bids/offers configurations and initial group structures on the social welfare improvement by the proposed group formation algorithm.}
\end{figure}


\subsubsection{An Example Allocation after Convergence}

Figure~\ref{fig:sim_convergence1_examplestructure} shows the group structure and instance allocation obtained at Iteration 20 with one initial group structure $\Pi = (G_0=\{1,2,3,4,5,6,7,8\}$, $G_1=\{\emptyset, \{1\}\}$, $G_2=\{\emptyset, \{2\}\} )$ and $L=\{(i,0)|i=1,\ldots,8\}$, denoted by a circle-pointed line in Figure~\ref{fig:sim_convergence_init}.
The final group structure becomes $\Pi = ( G_0=\{4\}$, $G_1=\{\{1,2,3,5,6,7,8\}, \{1,2\}\})$, and $L=\{(1,2),(2,2),(3,1),(4,0),$ $(5,1),(6,1),(7,2),(8,1)\}$.
There are some important observations in this result as follows:
\begin{itemize}
\item
In time slot 1, Users 3 and 8 are assigned to Provider 1 because Provider 1 offers lower price (\$.50) than Provider 2 (\$.60).
Users 1 and 2 are not assigned in time slot 1 because they can get more discount by being allocated after time slot 2 with other users. Users 1 and 2 cannot obtain the discount in time slot 1 with Users 3 and 8 because the total number of demands is 14, which does not affect price curves of providers.
\item
User 4 cannot be allocated by any providers. There are no sufficient supplies (e.g., 17 instances for Provider 1 and 19 instances for Provider 2) between time slots 2 and 5. That is, User 4 loses the auction due to relatively lower valuation (\$0.5 per-unit) than the others. 
\item
In time slots between 2 and 5, Provider 2 accepts Users 1, 2, 6, and 7 whose valuations are relatively higher than other users.
The valuation of Provider 2 becomes lower than that of Provider 1 when the number of total supplies exceeds 15, which results in the increase of social welfare.
User 3 with the highest valuation among all users remains to be assigned to Provider 1 because of the migration cost.
\item
User 8 is assigned in time slots between 1 and 6 (but not between 2 and 7). Given higher valuation of User 3 (\$1.5 per-unit) than that of User 6, the sum of utilities of Users 3 and 8 at time slot 1 is larger than that of Users 6 and 8 at time slot 7.
\item
Instances for User 6 are migrated from Provider 2 to 1 in time slot 7 due to the lower valuation of Provider 1.
\end{itemize}



\subsection{Optimality}
\label{subsec:eval_optimality}

We also verify that the proposed instance allocation, i.e., the outcome of Algorithm 1, is nearly optimal compared with an enumeration method, in which all possible allocation combinations are evaluated.
In addition, we compare the allocation in a different scheme called IA-FCFS, which stands for an Individual Auction by the First-Come-First-Serve allocation. In this scheme, users individually join an auction in a first-come-first-serve manner, and their demands are allocated to providers in the order of valuation from the lowest among providers with sufficient supplies.
We call GA-BCT (i.e., Group Auction is executed at the Bidding Closing Time) for our proposed instance allocation given by Algorithm 1.

Figure~\ref{fig:sim_optimality} shows the statistics (e.g., maximum, minimum, and quarter percentiles) of the social welfare (i.e., the total utilities of participants) over 100 randomly generated configurations of bids and offers based on the parameter setting described in Table~\ref{tab:sim_parametersetting1}.
The initial group structure is randomly generated.
The maximum number of iterations is set to 20.
The number in the box indicates the sum of the differences of the values from the ones obtained by the enumeration.

In Figure~\ref{fig:sim_optimality}(a), the statistics for GA-BCT are better than those for IA-FCFS. 
The main reason is that instances are efficiently allocated to users so that they get price discounts, which leads to the increasing social welfare.
It also shows that the maximum value for GA-BCT is the same as that from the enumeration although the percentile values are lower because the proposed allocation cannot find the optimal solution in some configurations.
We also investigate the impact of the numbers of users and providers on the results in the Figures~\ref{fig:sim_optimality}(b)-(d).
As the number of participants increases, the difference of values in GA-BCT and Brute-force increases.
GA-BCT might not find the optimal value; however, the statistics for GA-BCT are still nearly the same as the optimal and definitely higher than those for IA-FCFS.

\begin{figure}[t]
  \centering
	\resizebox{1.0\linewidth}{!}{\includegraphics{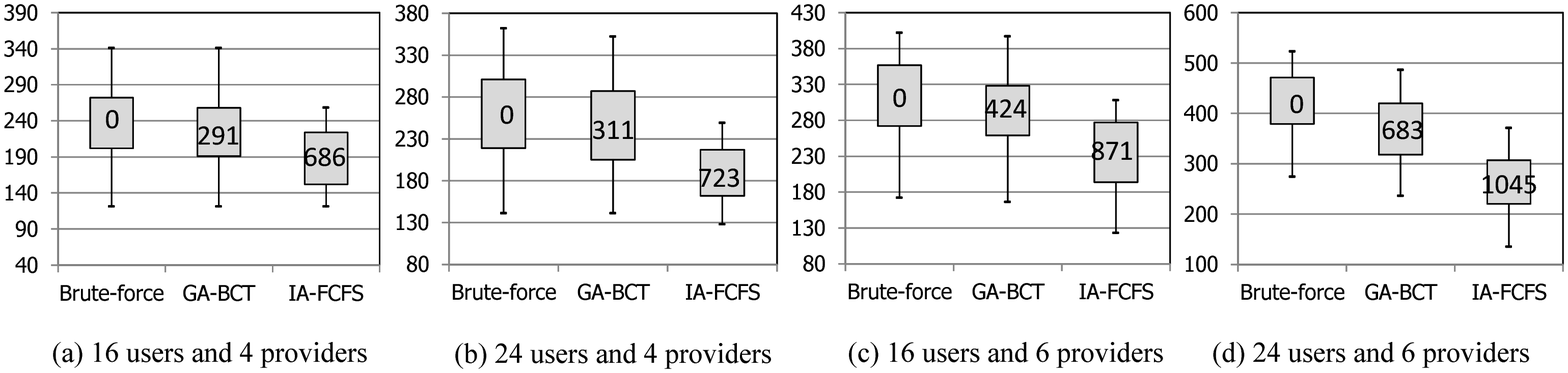}}
  \caption{The statistics of the social welfare over 100 randomly generated configurations. Y-axis indicates the social welfare given allocation, and the number in the box is the sum of the difference of the value from the one obtained by Brute-force.}
  \label{fig:sim_optimality}
\end{figure}

\subsection{Simulation Results}
\label{subsec:simulationresults}

This section presents a real-time simulation of the proposed group auction system and discusses its advantages in terms of the number of auction winners, resource utilization, and total profit of providers by showing comparison results with different schemes.

\subsubsection{Parameter setting for the simulation}

The simulation code is written in Matlab and run on a local machine (Intel Xeon 3.6GHz, 4G memory).
The input data is generated based on parameter settings described in Table~\ref{tab:sim_config}.
The obtained results vary based on the generated input data.
Because we do not have certain statistics (e.g., volume, period, requirement, etc) about user requests to the real cloud markets (e.g., SpotCloud, AmazonEC2, etc), we generate the input data with a uniform distribution.

We run the simulation for $SimT$=72 time slots period (e.g., 72 hours).
During the period, we fix 2 providers each of whom supplies 50 instances for each of 3 different types (i.e., K=3).
We assume that $n^b_{\tau}$, the number of new arrival users within time slot $\tau$, is randomly selected between 0 and 3.
Each user's bid is randomly generated.
Similar to the parameter setting described in the previous section, we set the maximum valuation-per-unit for instance type $k$ as 10$k$ cents and multiply the unit price by the number of instances and length that a user requests.
For example, for a user who requests instances of $\vec{d_i}=(8,5,3)$ for $\ell_i=3$, the value of his valuation is uniformly generated between 0 and $3\cdot(10\cdot8+20\cdot5+30\cdot3)=810$ cents (=\$8.10).
For the provider's valuation, price curves are randomly generated with $q_j^k[1] \in [5k, 10k]$ 
followed by $q_j^k[1] = \cdots = q_j^k[n-1] > q_j^k[n] = \cdots = q_j^k[s_j^k]$ where $n \in [2,s_j^k]$.
For the provider's valuation, we set price curves as shown in Figure~\ref{fig:example_pricecurve}.
For example, Provider 1's price curve for type $k=1$ is $q_1^1[1]=\$0.05$ and $q_1^1[31]=\$0.03$ while Provider 2's price curve is $q_2^1[1]=\$0.06$ and $q_2^1[16]=\$0.04$.


\begin{table}[t]
  \caption{Parameter Setting for the Simulation}
  \label{tab:sim_config}
\centering
\small
	\begin{tabular}{|l|c||l|c|}
	\hline
	Simulation time period & $SimT=72$ &
	\# of instance types	& $K=3$
  \\ \hline
	\# of arrival users within time slot $\tau$ & $n_{\tau}^b \in [0, 3]$ 	&
	\# of providers & $n^o=2$
	\\ \hline
	\multicolumn{2}{|c||}{User's bid $b_i$} &
	\multicolumn{2}{|c|}{Provider's offer $o_j$}
  \\ \hline
	Demand: \# of instances of type $k$  &  $d_i^k \in [0, 10]$ 	&
	Supply: \# of instances of type $k$  &  $s_j^k = 20$
	\\
	Demand Period: Length & $\ell_i \in [1, 6]$ 	&
	Supply Period: Length & $w_j = SimT$
	\\
	\hspace{+24mm} Starting time  &  $t_i^s \in$ [$\tau$+1, $\tau$+3] 	&
	\hspace{+22mm} Starting time  &  $t_j^s=1$
	\\
	\hspace{+24mm} Ending time & $t_i^e \in [t_i^s$+$\ell_i, \ t_i^s$+$\ell_i$+6] 	&
	\hspace{+22mm} Ending time & $t_j^e = SimT$
	\\
	Valuation & $v_i \sim $Uniform$(0,V)$ 	&
	Valuation & $Q_j = [\vec{q_j}^k]$
	\\
													 & $V$=$\ell_i \sum_{k \in K} 10k \cdot d_i^k$ 	&
													 & (Fig. \ref{fig:example_pricecurve})
	\\ \hline
	\end{tabular} 
\end{table}

\begin{figure}[t]
  \centering
	\resizebox{0.8\linewidth}{!}{\includegraphics{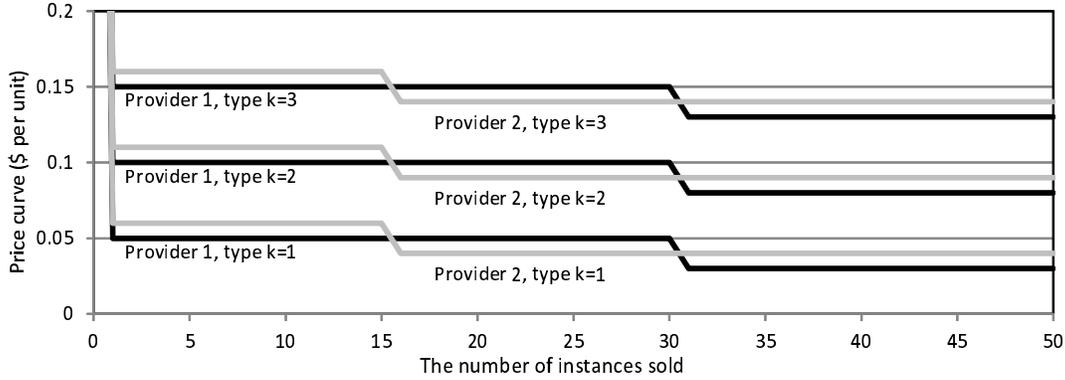}}
  \caption{Example price curves for Provider 1 and 2 for instance type of $k$=1, 2, and 3.}
  \label{fig:example_pricecurve}
\end{figure}

\subsubsection{Comparison}

In this section, we compare the proposed group auction scheme, called GA-BCT-GF (i.e., Group Auction is executed at the Bidding Closing Time in each of groups determined by Group Formation algorithm), with different schemes, i.e., IA-FCFS and GA-BCT.
Differences among the schemes are summarized in Table~\ref{tab:sim_scheme}.
The main differences are usages of Algorithms~\ref{algInstanceAllocation} and \ref{algGroupFormation}.
In IA-FCFS, both algorithms are not used; instead, users individually join an auction in a first-come-first-serve manner, and their demands are allocated to providers in the order of valuation from the lowest among providers with sufficient supplies.
In GA-BCT, whenever new participants arrive to the system, the instance allocation is computed by Algorithm~\ref{algInstanceAllocation} for all participants and executed at the bidding closing time.
In GA-BCT-GF, Algorithm~\ref{algGroupFormation} is used to construct a group formation. The instance allocation is computed by Algorithm~\ref{algInstanceAllocation} and executed at the bidding closing time for each of the groups.

\begin{table}[t]
  \caption{Description of different schemes}
  \label{tab:sim_scheme}
\centering
\small
	\begin{tabular}{|m{0.5in}|m{3in}|c|c|m{0.75in}|}
	\hline
	          &             & Dynamic       & Group         &             \\
	          & Description & Allocation    & Formation     & Cooperation \\
	          &             & (Algorithm 1) & (Algorithm 2) &             \\
	\hline \hline
	IA-FCFS   & Users individually join an auction in a first-come-first-serve manner, and their demands are allocated to providers in the order of valuation from the lowest among providers with sufficient supplies.
						& No & No & No participants cooperate. \\
	\hline
	GA-BCT    & All users and providers join a group auction where the instance allocation is computed by Algorithm 1 and executed at the bidding closing time.
						& Yes & No & All participants cooperate. \\
	\hline
	GA-BCT-GF & All users and providers are divided into multiple groups by Algorithm 2. The members in each of the groups join a group auction where the instance allocation is computed by Algorithm 1 and executed at the bidding closing time.
						& Yes & Yes & Participants cooperate in their own groups. \\
	\hline
	\end{tabular} 
\end{table}

Figure~\ref{fig:simresults1} shows the number of auction winners and losers during the simulation period in different schemes.
Table~\ref{tab:sim_comparison} summarizes the comparison results in terms of the total number of winners, resource utilization, profit of providers, and average payment of users.
From the number of winners, the request acceptance rate of users, which is the number of winners divided by the total number of users, will be obtained.
The resource utilization is given by the total allocated (i.e., provided) instances divided by the total supplies of providers.
The last column in the table indicates the improvement of the results by GA-BCT-GF compared to IA-FCFS.


The impact of the group bidding becomes stronger with a dynamic bid closing time.
In GA-BCT-GF, a central controller retains more bids before submitting a group bid. That is, the total number of requested instances becomes larger than that of IA-FCFS.
This results in providers to lower their offering prices, so more users have a chance to win the auction.
For example, during a time period between $t$=7 and 12 (indicated by a dotted box), 10 users in total join auctions.
Among the 10 users, 9 users win the auction in GA-BCT-GF while only 4 users win in IA-FCFS.
Similar results can be observed during a time period between $t$=31 and 39 (6 users win in GA-BCT-GF while 3 users win in IA-FCFS).
As shown in Table~\ref{tab:sim_comparison}, the total number of winners of GA-BCT-GF ($\gabctgfWin$ (81.1\%)) is larger than that of IA-FCFS ($\nbiafWin$ (71.6\%)).
GA-BCT-GF achieves \upReqAR of the improvement in the request acceptance rate.

The proposed group formation also impacts on the total number of winners.
It can be observed that users in GA-BCT-GF have more chance to win an auction before their deadline than those in GA-BCT. Since the GA-BCT considers one group of all cooperating participants, once the allocation is determined, it is executed at the bidding closing time of the group.
For example, at time $t$=9, 2 users lose the auction in GA-BCT (indicated by an arrow) while, in GA-BCT-GF, the group including the users did not execute the allocation immediately because the bidding closing time has not met yet, but at $t$=12.
Similar results can be observed at time $t$=17, 23, 30, and 40.
The total number of winners of GA-BCT-GF ($\gabctgfWin$) is slightly larger than that of GA-BCT ($\gabctWin$) as shown in Table~\ref{tab:sim_comparison}. 

In terms of the total profit, considering a group auction, Provider 2's profit increases from $\nbiafProvtwo$ (IA-FCFS) to $\gaProvtwo$ (GA-BCT) and $\gagfProvtwo$ (GA-BCT-GF) while Provider 1's profit decreases from $\nbiafProvone$ (IA-FCFS) to $\gaProvone$ (GA-BCT) and $\gagfProvone$ (GA-BCT-GF).
The reason is that when the number of requested instances is large (e.g., between 15 and 30 in Figure~\ref{fig:example_pricecurve}), Provider 2 offers cheaper prices than that of Provider 1. Consequently, the total number of provided instances by Provider 1 decreases, and hence, the profit slightly decreases.
However, if more users participates the auction, then Provider 1's profit increases again due to Provider 2's supply limitation.

During the simulation, 8640 instances (= 2 providers$\times$(20 instances$\times$3 types)$\times$72 SimTime) are supplied by providers, and 6090 instances ($\vec D$ = (\iaDone, \iaDtwo, \iaDthree)) for IA-FCFS, 6405 instances ($\vec D$ = \gaDone, \gaDtwo, \gaDthree)) for GA-BCT, and 7668 instances ($\vec D$ = \gagfDone, \gagfDtwo, \gagfDthree)) for GA-BCT-GF are allocated, respectively.
Due to the optimal instance allocation computation in GA-BCT and GA-BCT-GF, the resource utilization improves from $\nbiafResutil$ (=6090/8640) to $\gaResutil$ (=7155/8640) and $\gagfResutil$ (=7668/8640), which shows \upResUtil of the improvement.
It follows that the total profits of all providers increases from $\nbiafProvSum$ for IA-FCFS to $\gagfProvSum$ for GA-BCT-GF (i.e., $\upProfit$ of the improvement), which also gives the \upPay of decrements in the average payment (i.e., the average cost) for users from $\iaAvgCost$ for IA-FCFS to $\gaAvgCost$ for GA-BCT-GF.

\begin{figure}[t]
   \centering
   \subfigure{
      \includegraphics[scale=0.93]{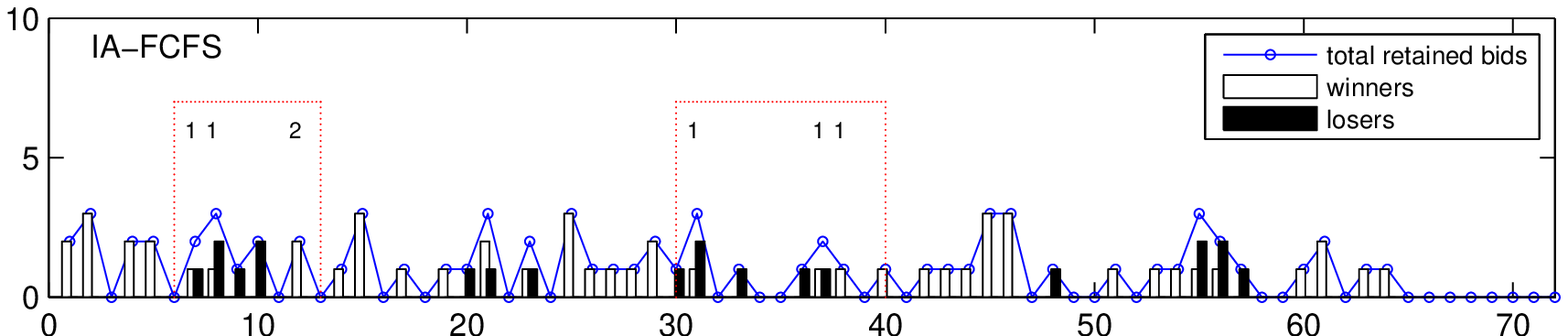}      
   }
   \subfigure{
      \includegraphics[scale=0.93]{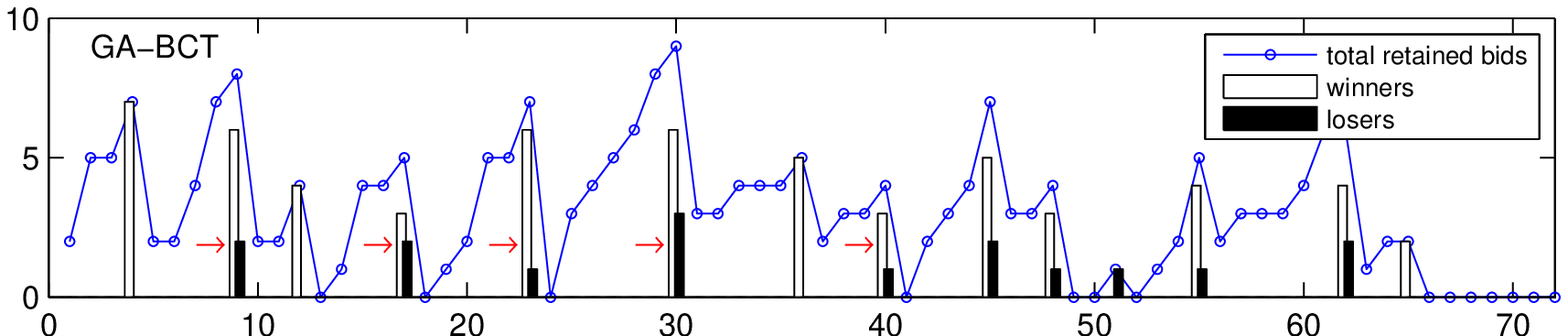}      
   }
   \subfigure{
      \includegraphics[scale=0.93]{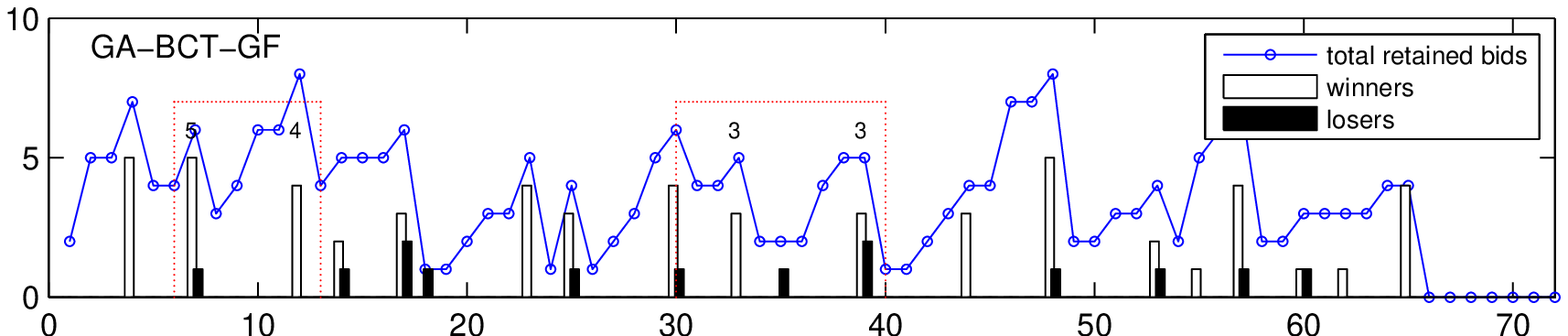}      
   }
	 \caption{The number of auction winners and losers.}
	 \label{fig:simresults1}
\end{figure}

\begin{table}[t]
  \caption{Summary of simulation results.}
  \label{tab:sim_comparison}
  \centering
	\begin{tabular}{|c|c||c|c|c|c|}
	\hline
	\multicolumn{2}{|c||}{} & \multicolumn{3}{|c|}{Scheme} & \multirow{2}{*}{Improvement}
	\\ \cline{3-5}
	\multicolumn{2}{|c||}{} & IA-FCFS & GA-BCT & GA-BCT-GF & 
	\\ \hline \hline
	\multicolumn{2}{|c||}{\#Winners (\#Losers)}
			& \nbiafWin (\nbiafLose) &   \gabctWin (\gabctLose)  & \gabctgfWin (\gabctgfLose) & \upReqAR $\uparrow$
	\\ \hline
  \multicolumn{2}{|c||}{Total \# of allocated}
  		& 6090 = (\iaDone,    &  7155 = (\gaDone,   & 7688 = (\gagfDone,    & \multirow{2}{*}{-}
  \\
  \multicolumn{2}{|c||}{instances ($D^1, D^2, D^3$)} 
  		& \iaDtwo, \iaDthree) &  \gaDtwo, \gaDthree)  & \gagfDtwo, \gagfDthree) &
  \\ \hline
  \multicolumn{2}{|c||}{Resource utilization}
  		& \nbiafResutil & \gaResutil  & \gagfResutil & \upResUtil $\uparrow$
  \\ \hline
  \multirow{3}{*}{\shortstack{Total \\ profit}}
      &  Provider 1 & \nbiafProvone & \gaProvone  & \gagfProvone &  \\
      &  Provider 2 & \nbiafProvtwo & \gaProvtwo  & \gagfProvtwo &  \\
      & [Sum] & [\nbiafProvSum] &   [\gaProvSum]  & [\gagfProvSum]  & \upProfit $\uparrow$
  \\ \hline
  \multicolumn{2}{|c||}{Avg. payment per user} 
  		& \iaAvgCost & \gaAvgCost   & \gagfAvgCost & \upPay $\downarrow$
  \\ \hline
	\end{tabular} 
\end{table}

Finally, we evaluate the system in the average case analysis.\footnote{In our simulation, different input data were generated in different simulation runs. So, it is hard to show the confidence interval of the results. If we run the simulation with the same set of bids and offers, then our system results in the almost same instance allocation. It means that the confidence interval is approximately 100\% but there is no meaning to us. This is one of the reasons to show the average case analysis rather than the confidence interval.}
Table~\ref{tab:sim_comparison_statistics} shows the statistics (e.g., average and standard deviation) of the results over 100 simulation runs based on the same parameter setting (Table 3).
The standard deviation is slightly high because the total number of users arrived during the simulation time period ($SimT=72$) is different in different simulation runs. 
However, GA-BCT-GF still outperforms the other schemes in this average case analysis.
GA-BCT and GA-BCT-GF accept more users (i.e., more winners) than IA-FCFS, which leads to increase the profit of providers.
Provider 2 gains more profit than that of Provider 1 when considering the group discount.
GA-BCT-GF produces more efficient instance allocation by considering a group formation than that of GA-BCT, which leads to improve the resource efficiency in terms of its utilization.

\begin{table}[h]
  \caption{Statistics (Avg $\pm$ Stddev) of the results over 100 simulation runs.}
  \label{tab:sim_comparison_statistics}
  \centering
	\begin{tabular}{|c|c||c|c|c|}
	\hline
	\multicolumn{2}{|c||}{} & \multicolumn{3}{|c|}{Scheme}
	\\ \cline{3-5}
	\multicolumn{2}{|c||}{} & IA-FCFS & GA-BCT & GA-BCT-GF
	\\ \hline \hline
	\multicolumn{2}{|c||}{Avg of \#Winners}
			& \CInbiafWin &   \CIgabctWin & \CIgabctgfWin
	\\ \hline
  \multicolumn{2}{|c||}{Avg of Resource utilization}
  		& \CInbiafResutil & \CIgaResutil  & \CIgagfResutil 
  \\ \hline
  \multirow{3}{*}{\shortstack{Avg of \\ Total \\ profit}}
      &  Provider 1 & \CInbiafProvone & \CIgaProvone  & \CIgagfProvone   \\
      &  Provider 2 & \CInbiafProvtwo & \CIgaProvtwo  & \CIgagfProvtwo   \\
      & [Sum] & [\CInbiafProvSum] &   [\CIgaProvSum]  & [\CIgagfProvSum]  
  \\ \hline
	\end{tabular} 
\end{table}

In summary, 
the proposed real-time group auction system with cooperation outperforms two different schemes, mimicking existing cloud hosting services, in terms of  the resource efficiency (e.g., \upReqAR higher request acceptance rate for users, obtained from the number of winners, and \upResUtil higher resource utilization for providers) and monetary benefits (e.g., \upPay lower average payments for users and \upProfit higher total profits for providers) with the considered simulation setting.
We have shown these advantages in the average case analysis.

\section{Conclusion}
\label{sec:conclusion}

This paper has proposed a real-time group auction system in the cloud instance market, and the simulation studies have verified its applicability and effectiveness in terms of resource efficiency and monetary benefits to auction participants.
We have shown that the proposed system outperforms two different schemes, mimicking existing cloud hosting services, in terms of resource efficiency and monetary benefits.
The proposed group formation and instance allocation algorithms have been analyzed in their complexity, stability, and optimality. The obtained results have shown that the algorithms with quadratic time complexity find nearly optimal group structures regardless of configurations and initial conditions.
%

As our future work, we have left a few issues that have not solved yet and interesting research directions as follows.
First, when to execute the proposed algorithms and when to close the auction are important problems in real-time auction because the computation of instance allocation and group formation have time overhead. In the current system, the bid closing time may not be the best. Finding the {\it best bid closing time} is challenging.
Second, we will improve the {\it scalability} of our algorithms to accommodate more participants in a reasonable time. In this paper, we have mainly focused on the design of the system model and algorithms, and we have shown the theoretical justification instead of evaluating a large number of participants.
Third, we will investigate various {\it bidding strategies} of participants and also different {\it price determination schemes} to know how they affect the monetary benefits to the participants. Those are pre-determined in the current system.
%
%
Finally, we will also consider our approach in various and complex situations.
For example, when a user requests multiple applications that have dependency, it might be interesting to investigate the best resource provisioning, i.e., where (or which machines) to locate which application, in terms of the performance such as response time and throughput.
It might be also interesting to consider the case that two users request the same application operated by the same provider, and the two users share the resources for the application.

%
%
%

\ifCLASSOPTIONcaptionsoff
  \newpage
\fi


\bibliographystyle{IEEEtran}
\bibliography{reference}

\begin{thebibliography}{10}
\providecommand{\url}[1]{#1}
\csname url@samestyle\endcsname
\providecommand{\newblock}{\relax}
\providecommand{\bibinfo}[2]{#2}
\providecommand{\BIBentrySTDinterwordspacing}{\spaceskip=0pt\relax}
\providecommand{\BIBentryALTinterwordstretchfactor}{4}
\providecommand{\BIBentryALTinterwordspacing}{\spaceskip=\fontdimen2\font plus
\BIBentryALTinterwordstretchfactor\fontdimen3\font minus
  \fontdimen4\font\relax}
\providecommand{\BIBforeignlanguage}[2]{{%
\expandafter\ifx\csname l@#1\endcsname\relax
\typeout{** WARNING: IEEEtran.bst: No hyphenation pattern has been}%
\typeout{** loaded for the language `#1'. Using the pattern for}%
\typeout{** the default language instead.}%
\else
\language=\csname l@#1\endcsname
\fi
#2}}
\providecommand{\BIBdecl}{\relax}
\BIBdecl

\bibitem{AmazonEC2}
``Amazon \uppercase{EC2},'' \url{http://aws.amazon.com/ec2/}.

\bibitem{eWinWin}
``e\uppercase{W}in\uppercase{W}in: A social buying technology company,''
  \url{http://ewinwin.com}.

\bibitem{groupgain}
``\uppercase{G}roup\uppercase{G}ain,'' \url{http://www.groupgain.com}.

\bibitem{groupon}
``\uppercase{GROUPON},'' \url{http://www.groupon.com}.

\bibitem{Voorsluys09}
W.~Voorsluys, J.~Broberg, S.~Venugopal, and R.~Buyya, ``Cost of virtual machine
  live migration in clouds: A performance evaluation,'' in \emph{Proc. of the
  1st International Conference on Cloud Computing}, 2009.

\bibitem{Berry10}
R.~Berry, M.~L. Honig, and R.~Vohra, ``Spectrum markets: motivation,
  challenges, and implications,'' \emph{IEEE Communications Magazine}, vol.
  48(11), November 2010.

\bibitem{Mutlu09}
H.~Mutlu, M.~Alanyali, and D.~Starobinski, ``Spot pricing of secondary spectrum
  access in wireless cellular networks,'' \emph{IEEE/ACM Transaction on
  Networking}, vol.~17, December 2009.

\bibitem{Markus04}
M.~Burger, B.~Klar, A.~Muller, and G.~Schindlmayr, ``A spot market model for
  pricing derivatives in electricity markets,'' \emph{Quantitative Finance},
  vol.~4, December 2004.

\bibitem{Kian05}
A.~R. Kian and J.~B. Cruz, Jr., ``Bidding strategies in dynamic electricity
  markets,'' \emph{Decision Support System}, vol. 40(3-4), October 2005.

\bibitem{Bhalgat11}
A.~Bhalgat, ``Online allocation of display ads with smooth delivery,'' in
  \emph{Proc. of the 7th Ad Auctions Workshop}, June 2011.

\bibitem{Mahdian07}
M.~Mahdian, H.~Nazerzadeh, and A.~Saberi, ``Allocating online advertisement
  space with unreliable estimates,'' in \emph{Proc. of the 8th ACM Conference
  on Electronic Commerce}, June 2007.

\bibitem{Wei10}
G.~Wei, A.~V. Vasilakos, Y.~Zheng, and N.~Xiong, ``A game-theoretic method of
  fair resource allocation for cloud computing services,'' \emph{The Journal of
  Supercomputing}, vol.~54, 2010.

\bibitem{Ardagna11}
D.~Ardagna, B.~Panicucci, and M.~Passacantando, ``A game theoretic formulation
  of the service provisioning problem in cloud systems,'' in \emph{Proc. of the
  20th International Conference on World wide web}, 2011.

\bibitem{Chaisiri11}
S.~Chaisiri, R.~Kaewpuang, B.~S. Lee, and D.~Niyato, ``Cost minimization for
  provisioning virtual servers in amazon elastic compute cloud,'' in
  \emph{Proc. of the International Symposium on Modeling, Analysis and
  Simulation, of Computer and Telecommunication Systems}, 2011.

\bibitem{Singh06}
G.~Singh, C.~Kesselman, and E.~Deelman, ``Application-level resource
  provisioning on the grid,'' in \emph{Proc. of the 2nd IEEE International
  Conference on e-Science and Grid Computing}, 2006.

\bibitem{Chimakurthi11}
L.~Chimakurthi and M.~K. SD, ``Power efficient resource allocation for clouds
  using ant colony framework,'' \emph{The Computing Research Repository}, vol.
  abs/1102.2608, 2011.

\bibitem{JBae08}
J.~Bae, E.~Beigman, R.~A. Berry, M.~L. Honig, and R.~V. Vohra, ``Sequential
  bandwidth and power auctions for distributed spectrum sharing,'' \emph{IEEE
  Journal on Communications}, vol.~26, September 2008.

\bibitem{Stanojev10}
I.~Stanojev, O.~Simeone, U.~Spagnolini, Y.~Bar-Ness, and R.~Pickholtz,
  ``Cooperative arq via auction-based spectrum leasing.'' \emph{IEEE
  Transactions on Communications}, vol.~58, 2010.

\bibitem{TanGurd07}
Z.~Tan and J.~R. Gurd, ``Market-based grid resource allocation using a stable
  continuous double auction,'' in \emph{Proc. of the 8th IEEE/ACM International
  Conference on Grid Computing}, 2007.

\bibitem{Fujiwara10}
I.~Fujiwara, K.~Aida, and I.~Ono, ``Applying double-sided combinational
  auctions to resource allocation in cloud computing,'' in \emph{Proc. of the
  10th IEEE/IPSJ Int'l Symposium on Applications and the Internet}, 2010.

\bibitem{Linli05}
L.~He and T.~R. Ioerger, ``Forming resource-sharing coalitions: a distributed
  resource allocation mechanism for self-interested agents in computational
  grids,'' in \emph{Proc. of the 2005 ACM symposium on Applied computing},
  March 2005.

\bibitem{GongYong03}
G.~Yong, Y.~Li, Z.~Wei-ming, S.~Ji-chang, and W.~Chang-ying, ``Methods for
  resource allocation via agent coalition formation in grid computing
  systems,'' in \emph{Proc. of 2003 IEEE International Conference on Robotics,
  Intelligent Systems and Signal Processing}, October 2003.

\bibitem{Pascual09}
F.~Pascual, K.~Rzadca, and D.~Trystram, ``Cooperation in multi-organization
  scheduling,'' \emph{Concurrency and Computation: Practice \& Experience},
  vol.~21, no.~7, pp. 905--921, May 2009.

\bibitem{AmazonEC2-SpotInstance}
``Amazon \uppercase{EC2} spot instance,''
  \url{http://aws.amazon.com/ec2/spot-instances/}.

\bibitem{SpotCloud}
``Spod\uppercase{C}loud,'' \url{http://www.spotcloud.com/}.

\bibitem{Urgaonkar07}
B.~Urgaonkar, G.~Pacifici, P.~Shenoy, M.~Spreitzer, and A.~Tantawi, ``Analytic
  modeling of multitier internet applications,'' \emph{ACM Transactions on the
  We}, vol.~1, no.~1, May 2007.

\bibitem{zaman2010}
S.~Zaman and D.~Grosu, ``Combinatorial auction-based allocation of virtual
  machine instances in clouds,'' in \emph{Proc. of the 2nd IEEE 2nd
  International Conference on Cloud Computing Technology and Science}, 2010.

\bibitem{Daniel2011}
D.~Warneke and O.~Kao, ``Exploiting dynamic resource allocation for efficient
  parallel data processing in the cloud,'' \emph{IEEE Transactions on Parallel
  and Distributed Systems}, vol.~22, 2011.

\bibitem{Bogomonlaia02}
A.~Bogomonlaia and M.~Jackson, ``The stability of hedonic coalition
  structures,'' \emph{Games and Economic Behavior}, vol.~38, January 2002.

\bibitem{DRay07}
D.~Ray, \emph{A Game-Theoretic Perspective on Coalition Formation}.\hskip 1em
  plus 0.5em minus 0.4em\relax Oxford University Press, 2007.

\end{thebibliography}


 

\begin{table}[h]
\begin{tabular}{cp{34.5em}}
\raisebox{-6.5em}{\includegraphics[width=1in,height=1.2in,clip,keepaspectratio]{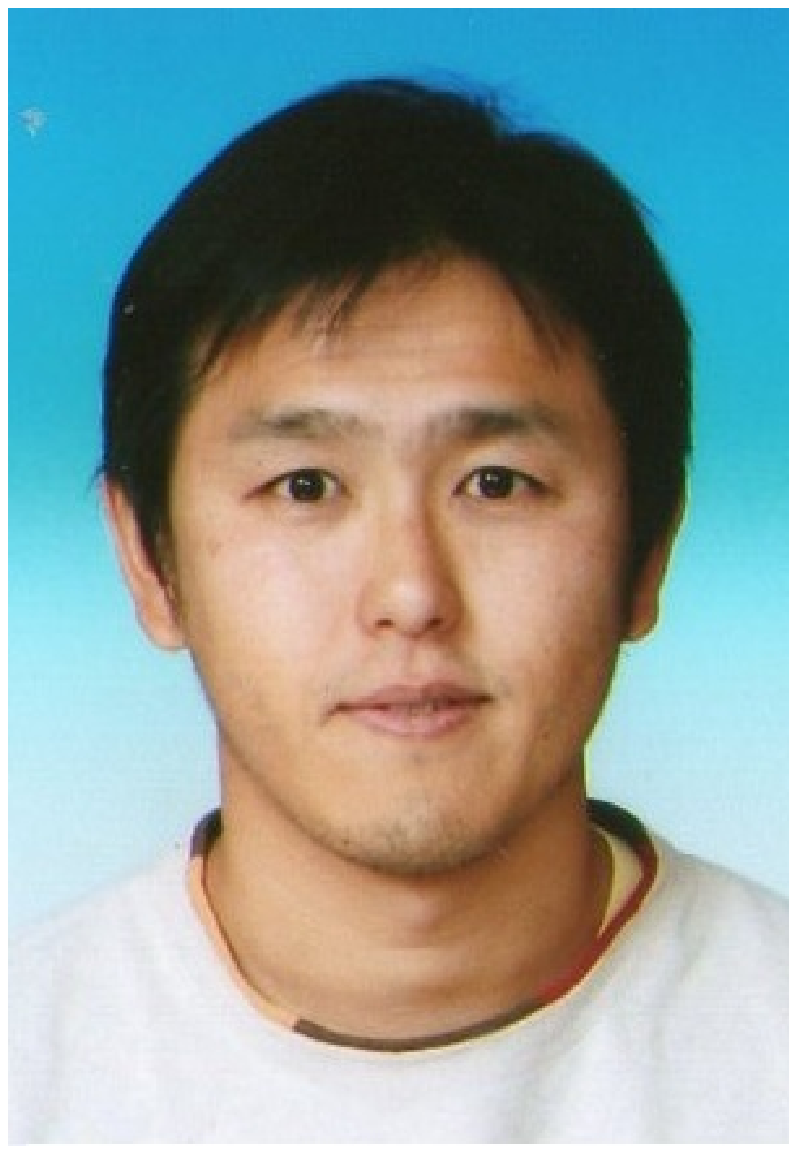}} &
\small
Chonho Lee currently works as a research fellow at the School of Computer Engineering, Nanyang Technological University, Singapore. He received his BS and MS in Computer Science from University of California, Irvine and Ph.D in Computer Science from University of Massachusetts, Boston. His current research interests include optimization and self-adaptation using game theory and bio-inspired approaches in large-scale network systems such as data centers and clouds.
\\ & \\
\raisebox{-7em}{\includegraphics[width=1in,height=1.2in,clip,keepaspectratio]{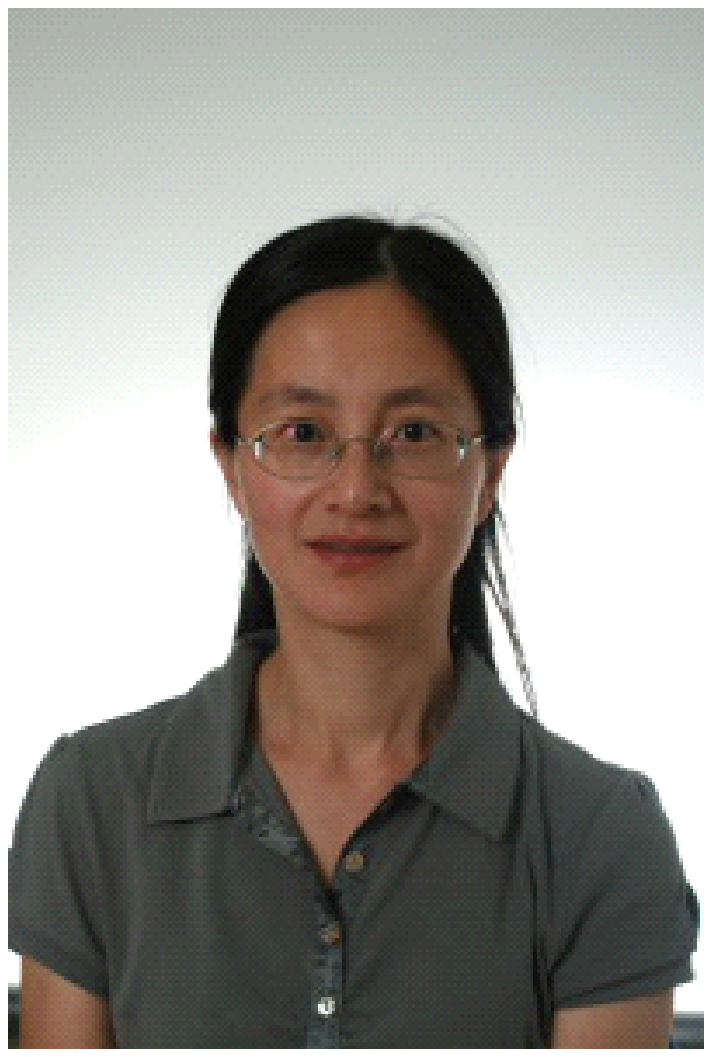}} &
\small
Ping Wang (M'08) received the PhD degree in electrical engineering from University of Waterloo, Canada, in 2008. Since June 2008, she has been an assistant professor in the School of Computer Engineering, Nanyang Technological University, Singapore. Her current research interests include resource allocation in multimedia wireless networks, cloud computing, and smart grid. She was a corecipient of the Best Paper Award from IEEE Wireless Communications and Networking Conference (WCNC) 2012 and IEEE International Conference on Communications (ICC) 2007. She is an Editor of IEEE Transactions on Wireless Communications, EURASIP Journal on Wireless Communications and Networking, and International Journal of Ultra Wideband Communications and Systems.
\\ & \\
\raisebox{-6.5em}{\includegraphics[width=1in,height=1.2in,clip,keepaspectratio]{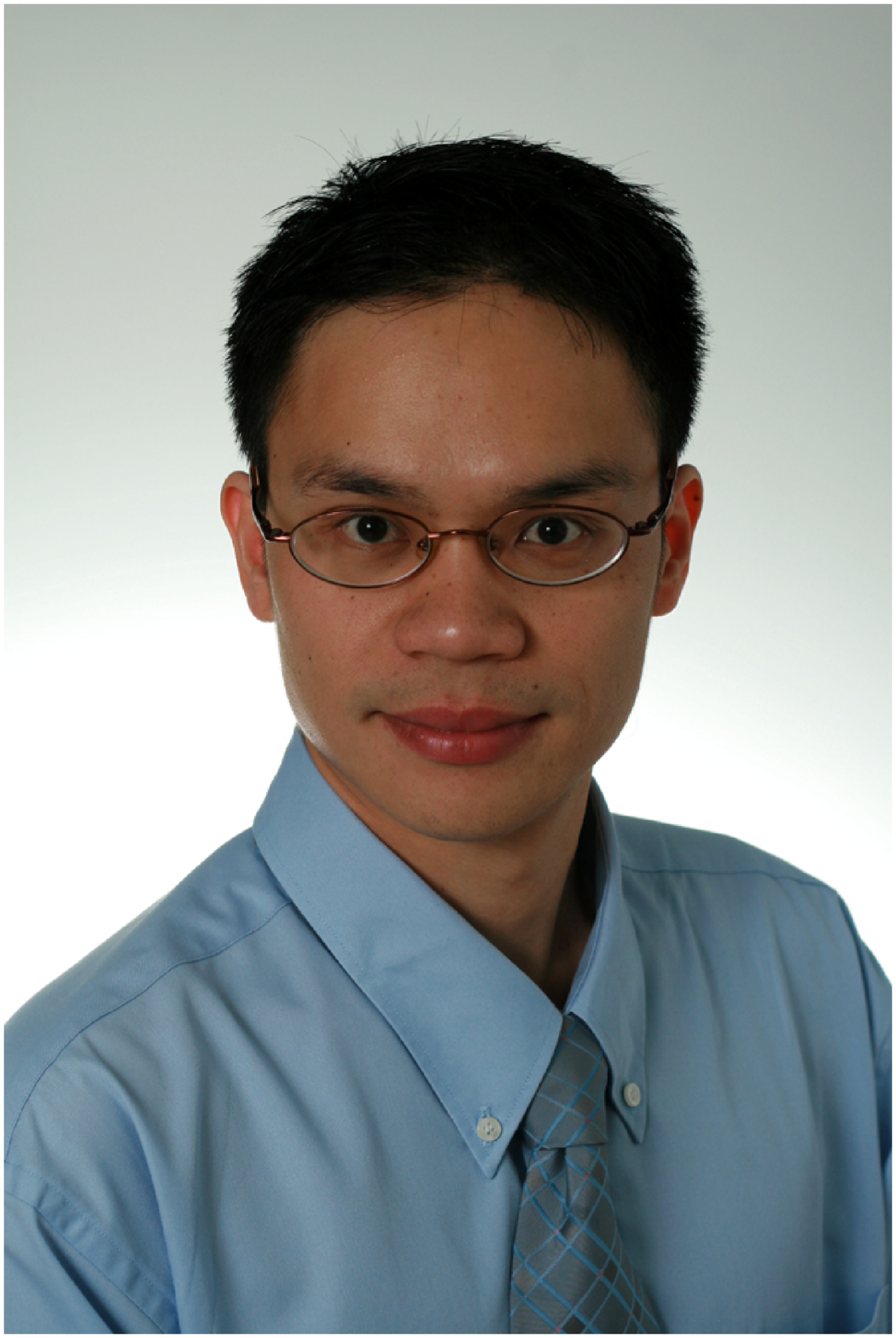}} &
\small
Dusit Niyato is currently an Assistant Professor in the School of Computer Engineering, at the Nanyang Technological University, Singapore. He obtained his Bachelor of Engineering in Computer Engineering from King Mongkut's Institute of Technology Ladkrabang (KMITL), Bangkok, Thailand. He received his Ph.D. in Electrical and Computer Engineering from the University of Manitoba, Canada.
\end{tabular}
\end{table}








\end{document}